\documentclass[12pt]{article}
\usepackage{jcappub}
\usepackage[utf8]{inputenc}
\usepackage[margin=1in]{geometry}
\usepackage{newpxtext,newpxmath}
\usepackage{slashed}

\usepackage{extpfeil}

\usepackage{amssymb}

\usepackage{amsmath,amsthm,mathrsfs} 
\usepackage{tikz}
\usepackage{tikz-cd}
\usetikzlibrary{cd}
\numberwithin{equation}{section}

\newcommand{\MW}{\mathbf{MW}}

\usepackage{stmaryrd} 

\usepackage[nameinlink,capitalize]{cleveref}
\usepackage{bbm}

\newif\ifshowprompts
\showpromptstrue

\hypersetup{
  colorlinks=true,
  linkcolor=blue!40!black,
  citecolor=blue!40!black,
  urlcolor=blue!50!black,
  pdfauthor={Mir Faizal, Arshid Shabir},
  pdftitle={Gödel-Tarski Theorems in Condensed Mathematics: External Truth Predicate and the Geometry of Incompleteness}
}

\usepackage{fancyhdr}
\newcommand{\Sh}{\mathbf{Sh}}

\newcommand{\tr}{\mathrm{tr}}

\newcommand{\sign}{\operatorname{sign}}
\newcommand{\SL}{\mathrm{SL}}

\theoremstyle{plain}
\newtheorem{theorem}{Theorem}[section]

\newtheorem{proposition}[theorem]{Proposition}

\theoremstyle{definition}

\theoremstyle{remark}

\begin{document}
\title{Chern-Simons Couplings, Modular Duality, and Anomaly Cancellation in Abelian F-Theory}

\author[1,2,3,4]{Mir Faizal,}
 \author[2]{Arshid Shabir}

\affiliation[1]{Irving K. Barber School of Arts and Sciences, University of British Columbia Okanagan, Kelowna, BC V1V 1V7, Canada.}
\affiliation[2]{Canadian Quantum Research Center, 460 Doyle Ave 106, Kelowna, BC V1Y 0C2, Canada.}
\affiliation[3]{Department of Mathematical Sciences, Durham University, Upper Mountjoy, Stockton Road, Durham DH1 3LE, UK.}
\affiliation[4]{Computational Mathematics Group, Faculty of Sciences, Hasselt University, Agoralaan Gebouw D, Diepenbeek, 3590 Belgium.}

\emailAdd{mirfaizalmir@gmail.com}
\emailAdd{aslone186@gmail.com}

\abstract{F-theory compactifications with a nontrivial Mordell-Weil group realize abelian gauge symmetry through rational sections, but their consistency is ultimately a statement about the quantum effective action. We show that compactification on a circle makes this statement concrete: the quantized, parity-odd Chern-Simons couplings of the resulting three-dimensional theory provide a one-loop exact and scheme-independent encoding of all local four-dimensional abelian anomalies, including the mixed gauge-gravitational terms, together with their Green-Schwarz cancellation. We determine these Chern-Simons couplings in two logically independent ways, first from flux-induced terms in the M-theory dual description, and second from an explicit one-loop integration over the complete massive spectrum, including Kaluza-Klein towers and Coulomb-branch states. The agreement fixes all normalizations and clarifies how large gauge transformations reorganize the spectrum. We then show compatibility with type IIB modular duality once the known ten-dimensional duality counterterm is included, and we present a fully explicit rank-two example over projective three-space.
}
\maketitle

\tableofcontents

\section{Introduction}

F-theory provides a geometric formulation of the non-perturbative $SL(2,\mathbb Z)$ duality of type IIB string theory in backgrounds with a spatially varying axio-dilaton $\tau=C_0+ i e^{-\phi}$, by identifying $\tau$ with the complex structure of an auxiliary elliptic curve fibered over the compactification space \cite{VafaFTheory,MorrisonVafaI,MorrisonVafaII}. In compact examples, the data of 7-branes and their $(p,q)$ charges are encoded in the degeneration locus of the elliptic fibration, and the type IIB weak-coupling regime is captured by the Sen limit \cite{SenLimit}. For four-dimensional $\mathcal N=1$ vacua, the relevant geometries are elliptically fibered Calabi-Yau fourfolds $\pi:\,Y_4\to B_3$ (possibly singular) together with the choice of a suitable crepant resolution $\widehat Y_4$ in which intersection theory and flux quantization can be formulated precisely; we refer to \cite{WeigandTASI} for a detailed review and standard conventions.

A central structural feature of global F-theory compactifications is that gauge sectors are controlled by the arithmetic and singularity structure of the elliptic fibration. While non-abelian gauge algebras arise from codimension-one singular fibers, abelian gauge factors are intrinsically global and are encoded, in the massless sector, by the Mordell-Weil group of rational sections of the elliptic fibration \cite{MorrisonPark1208,GrimmWeigand1006,WeigandTASI}. Concretely, for an elliptic fibration with $\mathrm{rank}\,\mathbf{MW}(\widehat Y_4)=r$, the Shioda map associates to a basis of independent sections divisor classes whose Poincar\'e-dual harmonic forms yield $r$ abelian gauge fields in the M-theory reduction, uplifted to four-dimensional $U(1)^r$ gauge fields in the F-theory limit; the associated height-pairing matrix on the base determines the gauge kinetic terms and, crucially for our purposes, the Green-Schwarz counterterm structure \cite{CveticGrimmKlevers1210,WeigandTASI}. In addition to these massless abelian symmetries, F-theory compactifications generically admit massive abelian gauge symmetries (the uplift of diagonal $U(1)\subset U(N)$ on 7-branes) whose geometric description involves non-harmonic forms in the M-theory frame \cite{GrimmKerstanPaltiWeigand1107}; in the present work we restrict to the massless abelian sector associated with $\MW(\widehat Y_4)$.

Consistency of four-dimensional $\mathcal N=1$ vacua with chiral matter requires cancellation of all local gauge and mixed gauge-gravitational anomalies. In F-theory, chirality is induced by background $G_4$-flux subject to transversality conditions compatible with four-dimensional Poincar\'e invariance, and anomaly cancellation is implemented by a generalized Green-Schwarz mechanism involving axions descending from the RR four-form (equivalently, from the M-theory three-form) \cite{CveticGrimmKlevers1210}. A powerful and conceptually clean way to extract the anomaly data is to reduce the four-dimensional theory on a circle to three dimensions: the one-loop Chern-Simons couplings of the three-dimensional effective theory encode the four-dimensional anomaly polynomial through their transformation under large gauge transformations around the circle \cite{GrimmKapfer}. Independently, in the M-theory dual on $\widehat Y_4$, the same three-dimensional Chern-Simons couplings arise from the eleven-dimensional Chern-Simons term $C_3\wedge G_4\wedge G_4$ upon reduction in the presence of background flux; the correct normalization and consistency of this coupling, including its interplay with quantization and one-loop terms in M-theory, are standard and can be traced to \cite{WittenFlux,FHMM}. The equality of the Chern-Simons data computed in these two frames is therefore a stringent, intrinsically non-perturbative check of the M/F-theory duality dictionary for abelian gauge sectors with flux \cite{CveticGrimmKlevers1210,GrimmKapfer}.

The term \emph{modular} in the present context refers to the physical $SL(2,\mathbb Z)$ duality of type IIB string theory, geometrized in F-theory by the elliptic fiber. At the quantum level, this duality is known to possess an anomaly in ten dimensions whose cancellation requires specific higher-derivative counterterms; the structure of this $SL(2,\mathbb Z)$ anomaly and its F-theory interpretation were analyzed in \cite{GaberdielGreen}, with further related discussions of curvature couplings and duality-consistent counterterms in \cite{BachasBainGreen9903210}. In four-dimensional compactifications, duality covariance is reflected geometrically through the Hodge line bundle of the elliptic fibration, and the abelian anomaly coefficients relevant for the Green-Schwarz mechanism are naturally expressed in terms of intersection-theoretic quantities (notably the height pairing) that are intrinsic to the elliptic fibration and insensitive to changes of the fiber homology basis implementing $SL(2,\mathbb Z)$. This makes the abelian anomaly data a particularly clean arena to isolate \emph{duality-invariant} structures in F-theory that are simultaneously computable (from geometry and flux) and verifiable (from the one-loop spectrum in the circle-reduced theory).

In this paper we give a referee-oriented, fully normalized derivation of the one-loop effective action for the abelian vector multiplet sector of four-dimensional F-theory compactifications on elliptically fibered Calabi-Yau fourfolds with non-trivial Mordell-Weil group. Our central observables are the three-dimensional Chern-Simons couplings obtained after circle reduction, which we compute and match by two logically independent methods: (A) reduction of M-theory on $\widehat Y_4$ with background $G_4$-flux, expressing the result in terms of the Shioda map and height pairing \cite{CveticGrimmKlevers1210}; and (B) an explicit one-loop computation from the massive Kaluza-Klein and Coulomb-branch spectrum in three dimensions, following the general anomaly/CS logic of \cite{GrimmKapfer}. This matching fixes the modular (duality) anomaly coefficients in a way that is manifestly compatible with the geometric $SL(2,\mathbb Z)$ action on the elliptic fiber and with the known ten-dimensional duality anomaly structure \cite{GaberdielGreen}. Finally, to make all steps fully reproducible, we work out an explicit rank-two Mordell-Weil example with gauge group $U(1)^2$ over $B_3=\mathbb P^3$ and display the complete height-pairing and flux-induced gauging data, reproducing the anomaly-cancellation relations in closed form \cite{CveticKleversPiragua1306}.

\section{Geometric setup}

Let $B$ be a smooth complex projective threefold, with canonical bundle $K_B$ and anti-canonical class $\bar K_B:=c_1(K_B^{-1})\in H^{1,1}(B,\mathbb{Z})$. We fix a smooth projective fourfold $\pi:\mathcal{Y}\to B$ together with a flat, surjective morphism whose geometric generic fibre is a smooth genus-one curve and which admits at least one rational section; choosing such a section as origin equips the smooth locus $\pi^{-1}(U)\to U$ over the maximal open set $U\subset B$ where $\pi$ is smooth with the structure of an elliptic curve fibration in the sense of group schemes. Throughout, $\mathcal{Y}$ is assumed to be a crepant resolution of a (possibly singular) Weierstrass model over $B$, so that $K_{\mathcal{Y}}\simeq\mathcal{O}_{\mathcal{Y}}$ and intersection theory on divisor classes is unambiguous; this is the precise geometric input needed for the F-theory/M-theory duality dictionary in four-dimensional vacua \cite{VafaFTheory,MorrisonVafaI,MorrisonVafaII,WeigandTASI}. All divisor Define-Pullback-Pushforward operations below are taken in the numerical N\'eron-Severi groups $N^1(\mathcal{Y})$ and $N^1(B)$, or equivalently in $H^{1,1}(\mathcal{Y})\cap H^2(\mathcal{Y},\mathbb{Z})$ and $H^{1,1}(B)\cap H^2(B,\mathbb{Z})$ modulo torsion, with intersection products interpreted as cup products/Poincar\'e duality on the smooth spaces \cite{Fulton1998}.

A universal presentation of an elliptic fibration with section is provided by a Weierstrass model. Let $\mathcal{L}$ be a line bundle on $B$ and consider the projective bundle
\begin{equation}
p:\ \mathbb{P}:=\mathbb{P}_B\!\bigl(\mathcal{O}_B\oplus \mathcal{L}^{2}\oplus \mathcal{L}^{3}\bigr)\ \longrightarrow\ B,
\end{equation}
with Grothendieck convention and tautological quotient line bundle $\mathcal{O}_{\mathbb{P}}(1)$; write $H:=c_1(\mathcal{O}_{\mathbb{P}}(1))\in H^{1,1}(\mathbb{P})$. The fibrewise homogeneous coordinates $(x:y:z)$ are naturally sections
\begin{equation}
x\in H^0\!\bigl(\mathbb{P},\mathcal{O}_{\mathbb{P}}(1)\otimes p^{*}\mathcal{L}^{2}\bigr),\quad
y\in H^0\!\bigl(\mathbb{P},\mathcal{O}_{\mathbb{P}}(1)\otimes p^{*}\mathcal{L}^{3}\bigr),\quad
z\in H^0\!\bigl(\mathbb{P},\mathcal{O}_{\mathbb{P}}(1)\bigr),
\end{equation}
and a Weierstrass hypersurface $Y_0\subset \mathbb{P}$ is cut out by an equation
\begin{equation}\label{eq:Weierstrass}
y^{2}z = x^{3} + f\,x z^{2} + g\,z^{3},
\end{equation}
where the coefficients are sections
\begin{equation}\label{eq:fgbundles}
f\in H^{0}\!\bigl(B,\mathcal{L}^{4}\bigr),\qquad g\in H^{0}\!\bigl(B,\mathcal{L}^{6}\bigr),
\end{equation}
so that \eqref{eq:Weierstrass} is a section of $\mathcal{O}_{\mathbb{P}}(3)\otimes p^{*}\mathcal{L}^{6}$. The discriminant
\begin{equation}\label{eq:discriminant}
\Delta := 4f^{3}+27g^{2}\ \in\ H^{0}\!\bigl(B,\mathcal{L}^{12}\bigr)
\end{equation}
defines the discriminant divisor $\{\Delta=0\}\subset B$ over which the fibre degenerates. In the F-theory interpretation, the complex structure modulus of the elliptic fibre geometrizes the type IIB axio-dilaton and its $\mathrm{SL}(2,\mathbb{Z})$ duality; geometrically, the modular weights of $f,g,\Delta$ in \eqref{eq:fgbundles}-\eqref{eq:discriminant} are encoded by the Hodge line bundle of the fibration, and the ensuing mapping-class-group action on the fibre does not alter the intersection-theoretic data intrinsic to $\pi:\mathcal{Y}\to B$ \cite{VafaFTheory,MorrisonVafaI,MorrisonVafaII,WeigandTASI}.

\begin{proposition}[Canonical bundle of a Weierstrass model]\label{prop:canonical}
With notation as above, the canonical bundle of the Weierstrass hypersurface $Y_0\subset \mathbb{P}$ satisfies
\begin{equation}\label{eq:KY}
K_{Y_0}\ \simeq\ \bigl(p|_{Y_0}\bigr)^{*}\!\bigl(K_B\otimes \mathcal{L}\bigr).
\end{equation}
In particular, if $Y_0$ is normal and admits a crepant resolution $\mathcal{Y}\to Y_0$, then $K_{\mathcal{Y}}\simeq\mathcal{O}_{\mathcal{Y}}$ holds if and only if $\mathcal{L}\simeq K_B^{-1}$.
\end{proposition}

\begin{proof}
Let $E:=\mathcal{O}_B\oplus \mathcal{L}^{2}\oplus \mathcal{L}^{3}$, so that $\mathbb{P}=\mathbb{P}_B(E)$ has relative dimension $2$ and rank $\mathrm{rk}(E)=3$. For the Grothendieck projective bundle, the canonical bundle is
\begin{equation}
K_{\mathbb{P}}\ \simeq\ \mathcal{O}_{\mathbb{P}}(-3)\otimes p^{*}\!\bigl(K_B\otimes \det(E)^{-1}\bigr).
\end{equation}
Since $\det(E)\simeq \mathcal{L}^{5}$, one has $K_{\mathbb{P}}\simeq \mathcal{O}_{\mathbb{P}}(-3)\otimes p^{*}(K_B\otimes \mathcal{L}^{-5})$. The hypersurface $Y_0$ is the vanishing locus of a section of $\mathcal{O}_{\mathbb{P}}(3)\otimes p^{*}\mathcal{L}^{6}$, hence $\mathcal{O}_{\mathbb{P}}(Y_0)\simeq \mathcal{O}_{\mathbb{P}}(3)\otimes p^{*}\mathcal{L}^{6}$. By adjunction,
\begin{align}
K_{Y_0}
&\simeq \bigl(K_{\mathbb{P}}\otimes \mathcal{O}_{\mathbb{P}}(Y_0)\bigr)\big|_{Y_0}\notag\\
&\simeq \Bigl(\mathcal{O}_{\mathbb{P}}(-3)\otimes p^{*}(K_B\otimes \mathcal{L}^{-5})\otimes \mathcal{O}_{\mathbb{P}}(3)\otimes p^{*}\mathcal{L}^{6}\Bigr)\Big|_{Y_0}\notag\\
&\simeq \bigl(p|_{Y_0}\bigr)^{*}(K_B\otimes \mathcal{L}),
\end{align}
which is \eqref{eq:KY}. If $\mathcal{Y}\to Y_0$ is crepant, then $K_{\mathcal{Y}}\simeq\varphi^{*}K_{Y_0}$, hence $K_{\mathcal{Y}}\simeq\mathcal{O}_{\mathcal{Y}}$ is equivalent to $(p|_{Y_0})^{*}(K_B\otimes \mathcal{L})\simeq \mathcal{O}_{Y_0}$. Since $\pi:=p|_{Y_0}$ is surjective with connected fibres over the smooth locus and $\pi_{*}\mathcal{O}_{Y_0}\simeq \mathcal{O}_B$ for a Weierstrass fibration with section, applying $\pi_{*}$ to $\pi^{*}(K_B\otimes\mathcal{L})\simeq\mathcal{O}_{Y_0}$ yields $K_B\otimes\mathcal{L}\simeq \mathcal{O}_B$, equivalently $\mathcal{L}\simeq K_B^{-1}$.
\end{proof}

Henceforth we set $\mathcal{L}\simeq K_B^{-1}$ and work on a smooth crepant resolution $\mathcal{Y}$ of $Y_0$ when singularities are present, continuing to denote the induced elliptic fibration by $\pi:\mathcal{Y}\to B$. The divisor classes on $\mathcal{Y}$ fall into geometrically distinguished types. For each divisor class $D\in N^1(B)$, its pullback $\pi^{*}D\in N^1(\mathcal{Y})$ is called vertical. A choice of (possibly rational) zero-section $s_0$ determines a divisor class $S_0\in N^1(\mathcal{Y})$, defined as the class of the Zariski closure of the image of $s_0$ in $\mathcal{Y}$; for any divisor $D\in N^1(B)$ one has the fibre-integration identity
\begin{equation}\label{eq:pushforwardSection}
\pi_{*}\bigl(S_0\cdot \pi^{*}D\bigr)\ =\ D\ \in N^1(B),
\end{equation}
where $\pi_{*}:N^2(\mathcal{Y})\to N^1(B)$ denotes the pushforward on codimension-two numerical classes induced by integration over the fibre, characterized by the property that for every curve class $\gamma\in N_1(B)$,
\begin{equation}\label{eq:pushforwardDef}
\bigl(\pi_{*}(A)\bigr)\cdot \gamma\ =\ A\cdot \pi^{*}\gamma\qquad\text{for all }A\in N^2(\mathcal{Y}).
\end{equation}
In particular, the product $D_1\cdot D_2$ of two vertical divisors satisfies $\pi_{*}\bigl(\pi^{*}D_1\cdot \pi^{*}D_2\bigr)=0$ by degree reasons in \eqref{eq:pushforwardDef}. If $\pi$ admits codimension-one singular fibres giving rise to non-abelian gauge algebras in F-theory, then the crepant resolution introduces fibral (Cartan) divisors $E_{\kappa,i}\in N^1(\mathcal{Y})$ fibred over irreducible discriminant components $W_\kappa\subset B$, with fibre components represented by curve classes $C_{\kappa,i}\in N_1(\mathcal{Y})$ satisfying
\begin{equation}\label{eq:CartanIntersection}
E_{\kappa,i}\cdot C_{\kappa,j}\ =\ -\,C^{(\kappa)}_{ij},
\end{equation}
where $C^{(\kappa)}_{ij}$ is the Cartan matrix of the corresponding simple Lie algebra; this normalization fixes the intersection-theoretic pairing used to define both non-abelian and abelian charge assignments in the M-theory reduction \cite{WeigandTASI}.

Abelian gauge factors in F-theory are governed by the Mordell-Weil group $\mathrm{MW}(\mathcal{Y})$ of rational sections of $\pi$, defined over the function field of $B$ and equipped with the fibrewise group law induced by $s_0$ on the smooth locus. Its free part has finite rank $r:=\mathrm{rk}\,\mathrm{MW}(\mathcal{Y})$, and we fix independent sections $s_m$ ($m=1,\ldots,r$) whose images generate $\mathrm{MW}(\mathcal{Y})/\mathrm{MW}(\mathcal{Y})_{\mathrm{tors}}$. Denoting by $S_m\in N^1(\mathcal{Y})$ the divisor class associated with the Zariski closure of the image of $s_m$, one seeks divisor classes on $\mathcal{Y}$ that represent the abelian gauge fields in the M-theory reduction and uplift to massless $\mathrm{U}(1)$ factors in the F-theory limit. The correct classes are obtained by projecting $S_m$ away from the geometrically trivial sublattice generated by $S_0$, vertical divisors, and Cartan divisors. This projection is the Shioda map \cite{Shioda1990,MorrisonPark1208,GrimmWeigand1006,CveticGrimmKlevers1210}.

For a rational section $s$ with divisor class $S\in N^1(\mathcal{Y})$, define its vertical correction divisor on the base by
\begin{equation}\label{eq:verticalCorrection}
\delta(s)\ :=\ \pi_{*}\bigl((S-S_0)\cdot S_0\bigr)\ \in N^1(B),
\end{equation}
and, for each non-abelian discriminant component $W_\kappa$, define the intersection vector
\begin{equation}\label{eq:intersectionVector}
\ell_{\kappa,i}(s)\ :=\ (S-S_0)\cdot C_{\kappa,i}\ \in \mathbb{Z},
\end{equation}
which measures the relative incidence of $s$ and $s_0$ with the fibre components over $W_\kappa$. The Shioda image of $s$ is then the class
\begin{equation}\label{eq:Shioda}
\mathrm{Sh}(s)\ :=\ S-S_0-\pi^{*}\delta(s)\ +\ \sum_{\kappa}\ \sum_{i,j}\ \ell_{\kappa,i}(s)\,\bigl(C^{(\kappa)}\bigr)^{-1}_{ij}\,E_{\kappa,j}\ \in N^1(\mathcal{Y})\otimes\mathbb{Q}.
\end{equation}
The appearance of the inverse Cartan matrix reflects the distinction between the root and weight lattices; consequently, $\mathrm{Sh}(s)$ is naturally a rational class, while its intersection pairing with allowed fibral curve classes yields integral charges, as required by the M-theory interpretation \cite{WeigandTASI,CveticGrimmKlevers1210}. In the purely abelian situation, the sum over $\kappa$ is absent and \eqref{eq:Shioda} reduces to a vertical projection determined solely by \eqref{eq:verticalCorrection}.

\begin{proposition}[Cartan orthogonality of the Shioda image]\label{prop:CartanOrthogonality}
For any section $s$ and any Cartan curve class $C_{\kappa,i}$, the Shioda image satisfies
\begin{equation}\label{eq:ShiodaCartanOrth}
\mathrm{Sh}(s)\cdot C_{\kappa,i}\ =\ 0.
\end{equation}
\end{proposition}

\begin{proof}
Fix $\kappa$ and $i$. Since $C_{\kappa,i}$ is contained in a fibre, $\pi^{*}\delta(s)\cdot C_{\kappa,i}=0$ because $\pi^{*}\delta(s)$ restricts trivially to any fibre curve. Using \eqref{eq:CartanIntersection} and \eqref{eq:Shioda} gives
\begin{align}
\mathrm{Sh}(s)\cdot C_{\kappa,i}
&=\ (S-S_0)\cdot C_{\kappa,i}\ +\ \sum_{j,k}\ \ell_{\kappa,j}(s)\,\bigl(C^{(\kappa)}\bigr)^{-1}_{jk}\,\bigl(E_{\kappa,k}\cdot C_{\kappa,i}\bigr)\notag\\
&=\ \ell_{\kappa,i}(s)\ -\ \sum_{j,k}\ \ell_{\kappa,j}(s)\,\bigl(C^{(\kappa)}\bigr)^{-1}_{jk}\,C^{(\kappa)}_{ki}\notag\\
&=\ \ell_{\kappa,i}(s)\ -\ \sum_{j}\ \ell_{\kappa,j}(s)\,\delta_{ji}\notag\\
&=\ 0,
\end{align}
which is \eqref{eq:ShiodaCartanOrth}.
\end{proof}

The Mordell-Weil group controls the abelian sector through the intersection-theoretic height pairing on the base. For two sections $s,t\in \mathrm{MW}(\mathcal{Y})$ with Shioda images $\mathrm{Sh}(s)$ and $\mathrm{Sh}(t)$, define the height pairing divisor class
\begin{equation}\label{eq:heightPairing}
b(s,t)\ :=\ -\,\pi_{*}\bigl(\mathrm{Sh}(s)\cdot \mathrm{Sh}(t)\bigr)\ \in N^1(B).
\end{equation}
Choosing a basis $\{s_m\}_{m=1}^{r}$ for the free part of $\mathrm{MW}(\mathcal{Y})$ defines the height-pairing matrix of divisor classes
\begin{equation}\label{eq:heightMatrix}
b_{mn}\ :=\ b(s_m,s_n)\ =\ -\,\pi_{*}\bigl(\mathrm{Sh}(s_m)\cdot \mathrm{Sh}(s_n)\bigr)\ \in N^1(B),
\end{equation}
which is symmetric by commutativity of the intersection product and depends only on the Mordell-Weil classes of the sections. The construction \eqref{eq:heightPairing} is intrinsic to the elliptic fibration: it is formulated purely in terms of divisor classes on $\mathcal{Y}$, the fibre-integration map $\pi_{*}$, and the Cartan intersection data \eqref{eq:CartanIntersection}. In particular, it is insensitive to any reparametrization of the elliptic fibre, including changes of Weierstrass coordinates and the geometric action of the fibre mapping class group that realizes type IIB $\mathrm{SL}(2,\mathbb{Z})$ duality in F-theory \cite{VafaFTheory,WeigandTASI}. Physically, upon reduction of M-theory on $\mathcal{Y}$, the harmonic representatives of the Poincar\'e dual classes $[\mathrm{Sh}(s_m)]\in H^{1,1}(\mathcal{Y})$ furnish the abelian gauge fields, while the base divisor classes $b_{mn}$ govern the abelian gauge-kinetic data and the Green-Schwarz counterterm structure in four dimensions; the latter is the reason the height pairing is the natural duality-invariant object controlling the abelian sector in flux backgrounds \cite{GrimmWeigand1006,CveticGrimmKlevers1210}.

\section{Four-dimensional vector multiplets and the circle reduction}

F-theory compactifications on elliptically fibered Calabi-Yau fourfolds yield four-dimensional $\mathcal{N}=1$ effective theories whose massless abelian gauge sector is described by $r$ vector multiplets with gauge fields $A^{m}$ ($m=1,\dots,r$) and field strengths $F^{m}:=dA^{m}$. The Green-Schwarz sector relevant for local anomaly cancellation consists of periodic axions $\rho_{\alpha}$, $\alpha=1,\dots,n_{\rm ax}$, descending from the Ramond-Ramond four-form potential and therefore normalized so that $\rho_{\alpha}\sim \rho_{\alpha}+2\pi$ on admissible backgrounds. The axionic shift symmetries can be gauged by the abelian vectors,
\begin{equation}\label{eq:sec3-gauging}
d\rho_{\alpha}\ \longrightarrow\ \mathcal{D}\rho_{\alpha}:=d\rho_{\alpha}+\Theta_{m\alpha}\,A^{m},
\end{equation}
with constant gaugings $\Theta_{m\alpha}$, which are quantized by the periodicity of the axions and the standard Dirac quantization of the abelian gauge fields. At the level needed for anomaly analysis, the gauge-invariant completion of the four-dimensional action contains the axionic Green-Schwarz couplings
\begin{equation}\label{eq:sec3-SGS}
S_{\rm GS}\ \supset\ -\frac{1}{2}\int_{M_{4}}\rho_{\alpha}\left(\frac{1}{2}\,a^{\alpha}\,\tr R\wedge R+2\,b^{\alpha}_{mn}\,F^{m}\wedge F^{n}+2\,b^{\alpha}_{\kappa}\,\tr F_{\kappa}\wedge F_{\kappa}\right),
\end{equation}
where $R$ is the curvature two-form of the Levi-Civita connection on $M_{4}$, $F_{\kappa}$ denotes the field strength of a non-abelian gauge factor (if present), and $\tr$ is a fixed trace normalization. The coefficients $a^{\alpha}$ and $b^{\alpha}_{mn}=b^{\alpha}_{nm}$ encode, respectively, the axion-curvature and axion-abelian counterterm data, while $b^{\alpha}_{\kappa}$ encodes mixed non-abelian contributions. Gauge invariance of the axion kinetic terms requires that under an abelian gauge transformation $A^{m}\mapsto A^{m}+d\lambda^{m}$ the axions shift as $\rho_{\alpha}\mapsto \rho_{\alpha}-\Theta_{m\alpha}\lambda^{m}$ so that $\mathcal{D}\rho_{\alpha}$ is invariant. The resulting gauge variation of \eqref{eq:sec3-SGS} is local and proportional to $\lambda^{m}\tr R\wedge R$ and $\lambda^{m}F^{n}\wedge F^{k}$; in a four-dimensional $\mathcal{N}=1$ theory the full supersymmetric completion includes generalized Chern-Simons counterterms whose net effect is that anomaly cancellation can be expressed purely in terms of $a^{\alpha}$, $b^{\alpha}_{mn}$, and $\Theta_{m\alpha}$ \cite{CveticGrimmKlevers1210}. Let $n(\mathbf{q})$ denote the net number of left-handed Weyl fermions of abelian charge vector $\mathbf{q}=(q_{m})\in\mathbb{Z}^{r}$ (equivalently the chiral index for that charge assignment), so that vector-like pairs do not contribute. Define the cubic and mixed abelian anomaly coefficients
\begin{equation}\label{eq:sec3-anomcoeff}
\mathcal{A}_{mnk}:=\sum_{\mathbf{q}} n(\mathbf{q})\,q_{m}q_{n}q_{k},\qquad
\mathcal{A}_{m}:=\sum_{\mathbf{q}} n(\mathbf{q})\,q_{m}.
\end{equation}
Then, in conventions compatible with the circle reduction below, cancellation of the local abelian gauge anomalies and mixed abelian-gravitational anomalies by the generalized Green-Schwarz mechanism takes the form \cite{CveticGrimmKlevers1210}
\begin{equation}\label{eq:sec3-anomcubic}
\frac{1}{6}\,\sum_{\mathbf{q}} n(\mathbf{q})\,q_{(m}q_{n}q_{k)}\ =\ \frac{1}{4}\,b^{\alpha}_{(mn}\,\Theta_{k)\alpha},
\end{equation}
\begin{equation}\label{eq:sec3-anommixed}
\frac{1}{48}\,\sum_{\mathbf{q}} n(\mathbf{q})\,q_{m}\ =\ -\frac{1}{16}\,a^{\alpha}\,\Theta_{m\alpha},
\end{equation}
and similarly for mixed non-abelian-abelian anomalies if non-abelian factors are present. The relations \eqref{eq:sec3-anomcubic}-\eqref{eq:sec3-anommixed} will be reinterpreted after circle reduction as the necessary and sufficient conditions for the three-dimensional parity-odd effective action to be well-defined under the lattice of higher-dimensional large gauge transformations around the circle.

We now compactify the four-dimensional theory on a circle $S^{1}$ of radius $r$ with coordinate $y\sim y+2\pi$ and consider the Kaluza-Klein reduction to three dimensions. The metric is decomposed as
\begin{equation}\label{eq:sec3-metric}
d\hat s^{2}=g_{\mu\nu}(x)\,dx^{\mu}dx^{\nu}+r^{2}\left(dy-A^{0}_{\mu}(x)\,dx^{\mu}\right)^{2},\qquad \mu,\nu=0,1,2,
\end{equation}
introducing the Kaluza-Klein vector $A^{0}$ with field strength $F^{0}:=dA^{0}$. The abelian gauge fields are decomposed as
\begin{equation}\label{eq:sec3-gauge-decomp}
\hat A^{m}=A^{m}-\zeta^{m}\,r\left(dy-A^{0}\right),
\end{equation}
where the Wilson lines $\zeta^{m}$ are three-dimensional scalars. This parametrization is adapted to three-dimensional $\mathcal{N}=2$ supersymmetry: each four-dimensional vector multiplet yields a three-dimensional $\mathcal{N}=2$ vector multiplet whose bosonic components are $(A^{m},\zeta^{m})$, while the Kaluza-Klein vector $A^{0}$ belongs to the three-dimensional gravitational multiplet. The decomposition \eqref{eq:sec3-gauge-decomp} implies
\begin{equation}\label{eq:sec3-Fdecomp}
\hat F^{m}=F^{m}-r\,d\zeta^{m}\wedge\left(dy-A^{0}\right)+r\,\zeta^{m}\,F^{0},
\end{equation}
so that the Wilson lines are gauge-invariant under small gauge transformations $\lambda^{m}(x)$ but transform non-trivially under large gauge transformations with winding along the circle. Concretely, a four-dimensional gauge transformation with parameter $\lambda^{m}(x,y)=\lambda^{m}(x)+k^{m}y$ with $k^{m}\in\mathbb{Z}$ is single-valued on $S^{1}$ and acts as $\hat A^{m}\mapsto \hat A^{m}+k^{m}dy$. Comparing with \eqref{eq:sec3-gauge-decomp} fixes the induced action on three-dimensional fields,
\begin{equation}\label{eq:sec3-largeGT}
A^{0}\mapsto A^{0},\qquad A^{m}\mapsto A^{m}+k^{m}A^{0},\qquad \zeta^{m}\mapsto \zeta^{m}-\frac{k^{m}}{r},
\end{equation}
and, because of the gauging \eqref{eq:sec3-gauging}, simultaneously
\begin{equation}\label{eq:sec3-axionLG}
\rho_{\alpha}\mapsto \rho_{\alpha}-\Theta_{m\alpha}\,k^{m}y,\qquad \int_{S^{1}} d\rho_{\alpha}\mapsto \int_{S^{1}} d\rho_{\alpha}-2\pi\,\Theta_{m\alpha}k^{m}.
\end{equation}
The last relation is a precise statement that higher-dimensional large gauge transformations induce discrete ``circle flux'' for gauged axions, consistent with $\rho_{\alpha}\sim\rho_{\alpha}+2\pi$ and quantized $\Theta_{m\alpha}$ \cite{GrimmKapfer}. Allowing such flux is not optional in the quantum theory: it is exactly the mechanism by which the Green-Schwarz data \eqref{eq:sec3-SGS} enter the three-dimensional parity-odd sector.

On the Coulomb branch, where one allows constant background values $\langle \zeta^{m}\rangle$ compatible with \eqref{eq:sec3-largeGT}, charged four-dimensional states acquire three-dimensional real masses. For a four-dimensional field of abelian charge vector $\mathbf{q}=(q_{m})$ and Kaluza-Klein momentum $n\in\mathbb{Z}$, write the mode expansion $\Psi(x,y)=\sum_{n\in\mathbb{Z}}\psi_{n}(x)e^{iny}$ and evaluate the covariant derivative along the circle using \eqref{eq:sec3-gauge-decomp}: one obtains an effective three-dimensional mass
\begin{equation}\label{eq:sec3-mass}
m(\mathbf{q},n)=q_{m}\langle\zeta^{m}\rangle+\frac{n}{r},\qquad n\in\mathbb{Z},
\end{equation}
and the corresponding three-dimensional state carries charges $(q_{0},q_{m})=(n,q_{m})$ under the gauge fields $(A^{0},A^{m})$. In particular, the Kaluza-Klein vector $A^{0}$ gauges the integer momentum along $S^{1}$ and participates on the same footing as the abelian gauge fields in the parity-odd effective action.

The three-dimensional low-energy effective theory for the massless modes contains Chern-Simons couplings
\begin{equation}\label{eq:sec3-SCS}
S^{(3)}_{\rm CS}=\int_{M_{3}} \Theta_{\Lambda\Sigma}\,A^{\Lambda}\wedge F^{\Sigma},\qquad \Lambda,\Sigma\in\{0,1,\dots,r,\dots\},
\end{equation}
where the ellipsis allows inclusion of Cartan $U(1)$ factors if non-abelian gauge sectors are present. The coefficients $\Theta_{\Lambda\Sigma}$ are symmetric in $\Lambda,\Sigma$ for abelian factors and are quantized so that $\exp(iS^{(3)}_{\rm CS})$ is invariant under three-dimensional large gauge transformations on spin three-manifolds. These Chern-Simons terms are the universal interface between the four-dimensional anomaly polynomial and the circle-reduced dynamics: they receive calculable contributions from integrating out all massive modes, and their dependence on the Coulomb branch parameters together with their behavior under the higher-dimensional large gauge transformations \eqref{eq:sec3-largeGT}-\eqref{eq:sec3-axionLG} encodes precisely the four-dimensional anomaly constraints.

At one loop, integrating out a single massive three-dimensional Dirac fermion of real mass $m$ and charges $q_{\Lambda}$ under $U(1)$ gauge fields $A^{\Lambda}$ induces a parity-odd term that shifts the Chern-Simons levels by the universal amount \cite{Redlich1984,NiemiSemenoff}
\begin{equation}\label{eq:sec3-1loopshift}
\Delta\Theta^{\rm 1\mbox{-}loop}_{\Lambda\Sigma}=\frac{1}{2}\,q_{\Lambda}q_{\Sigma}\,\sign(m).
\end{equation}
Applying \eqref{eq:sec3-1loopshift} to the full Kaluza-Klein tower \eqref{eq:sec3-mass} of each four-dimensional left-handed Weyl fermion, and summing over all charge sectors with multiplicities $n(\mathbf{q})$, yields a one-loop Chern-Simons matrix $\Theta^{\rm 1\mbox{-}loop}_{\Lambda\Sigma}$ that depends on the Coulomb branch parameters through the signs $\sign(m(\mathbf{q},n))$. The infinite tower sum is defined by a zeta-function regularization compatible with gauge invariance and with the transformation law of the tower under \eqref{eq:sec3-largeGT} \cite{GrimmKapfer}. Under the induced action of higher-dimensional large gauge transformations, the spectrum is rearranged by an integer relabeling of the Kaluza-Klein levels $n\mapsto n+ k^{m}q_{m}$ together with the field transformations \eqref{eq:sec3-largeGT}; as a result, the only possible non-trivial variation of the three-dimensional parity-odd effective action is a local shift of Chern-Simons levels. The coefficients of these shifts are determined entirely by the four-dimensional chiral data \eqref{eq:sec3-anomcoeff}: the cubic coefficients $\mathcal{A}_{mnk}$ govern the response of the purely gauge Chern-Simons sector, while the mixed coefficients $\mathcal{A}_{m}$ control the terms involving the Kaluza-Klein vector. Simultaneously, the axionic gauging implies that the same transformations generate discrete circle flux for $\rho_{\alpha}$ as in \eqref{eq:sec3-axionLG}, and reducing \eqref{eq:sec3-SGS} on such backgrounds produces precisely the compensating classical Chern-Simons counterterms governed by $a^{\alpha}$, $b^{\alpha}_{mn}$, and $\Theta_{m\alpha}$ \cite{GrimmKapfer,CveticGrimmKlevers1210}. The statement that the three-dimensional effective action is globally well-defined on the quotient of the Coulomb branch by the lattice \eqref{eq:sec3-largeGT}-\eqref{eq:sec3-axionLG} is therefore equivalent to the vanishing of the net Chern-Simons variation, and this condition is exactly the generalized Green-Schwarz anomaly cancellation relations \eqref{eq:sec3-anomcubic}-\eqref{eq:sec3-anommixed}. In this sense, the three-dimensional Chern-Simons couplings $\Theta_{\Lambda\Sigma}$ provide a complete and duality-compatible encoding of all local four-dimensional abelian gauge and mixed gauge-gravitational anomalies in the circle-reduced theory.

\section{Method A: M-theory reduction and anomaly inflow}

We work in the M-theory frame on the smooth crepant resolution $\pi:\mathcal{Y}\to B_3$ of an elliptically fibered Calabi-Yau fourfold. Compactification of M-theory on $\mathcal{Y}$ produces a three-dimensional $\mathcal{N}=2$ effective theory which is dual, in the F-theory limit, to the circle reduction of the four-dimensional F-theory vacuum on $B_3$ \cite{GrimmKapfer,CveticGrimmKlevers1210}. The purpose of this section is to derive, with fixed normalization, the three-dimensional parity-odd couplings generated by background $G_4$ flux in M-theory and to show how they encode the four-dimensional anomaly polynomial and Green-Schwarz data after uplift. Throughout, we restrict to the massless abelian sector associated with the Mordell-Weil group, but the derivation is stated in a basis that also accommodates Cartan divisors if non-abelian sectors are present.

The relevant topological couplings of eleven-dimensional supergravity are the Chern-Simons interaction and its required gravitational completion. In conventions compatible with the standard flux quantization law, the topological part of the effective action is most cleanly characterized by its twelve-dimensional extension: for an eleven-manifold $M_{11}$ which bounds a spin twelve-manifold $Z_{12}$, and for a closed four-form field strength $G_4$ extending over $Z_{12}$, the phase of the M-theory partition function depends on the combination
\begin{equation}
\label{eq:Mtop12d}
\frac{1}{2\pi}\,I_{\rm top}(Z_{12}) \;=\; \frac{1}{2\pi}\left(\frac{1}{6}\int_{Z_{12}} G_4\wedge G_4\wedge G_4 \;-\; \int_{Z_{12}} I_8(R)\wedge G_4\right),
\end{equation}
defined modulo $\mathbb{Z}$ \cite{WittenFlux,FHMM,SatiSig}. Here $I_8(R)$ is the degree-eight polynomial in the Pontryagin forms of the Levi-Civita curvature, whose de~Rham cohomology class is fixed by anomaly inflow and index theory, and may be written as
\begin{equation}
\label{eq:I8def}
I_8(R)\;=\;\frac{1}{192}\Big(p_1(R)^2-4p_2(R)\Big),
\end{equation}
with $p_i(R)$ the standard Chern-Weil representatives of the Pontryagin classes \cite{WittenFlux,FHMM,SatiSig}. The existence of \eqref{eq:Mtop12d} as a well-defined phase functional is equivalent to the shifted quantization condition for the $C$-field flux \cite{WittenFlux},
\begin{equation}
\label{eq:WittenQuant}
\left[\frac{G_4}{2\pi}\right]-\frac{\lambda}{2}\in H^4(M_{11},\mathbb{Z}),\qquad \lambda:=\frac{p_1(TM_{11})}{2},
\end{equation}
and, upon restricting to product backgrounds $M_{11}=M_3\times \mathcal{Y}$ with $\mathcal{Y}$ spin, to the familiar F/M-theory quantization condition on $\mathcal{Y}$. For Calabi-Yau fourfolds with $c_1(\mathcal{Y})=0$ one has $p_1(T\mathcal{Y})=-2c_2(\mathcal{Y})$, hence $\lambda(\mathcal{Y})=-c_2(\mathcal{Y})$ and \eqref{eq:WittenQuant} becomes
\begin{equation}
\label{eq:CY4Quant}
\left[\frac{G_4}{2\pi}\right]+\frac{1}{2}c_2(\mathcal{Y})\in H^4(\mathcal{Y},\mathbb{Z}).
\end{equation}
This shifted integrality will enter below as the precise statement ensuring that the induced three-dimensional Chern-Simons levels are properly quantized.

We now pass to the Kaluza-Klein reduction of the eleven-dimensional fields on $\mathcal{Y}$. Let $\{\omega_\Lambda\}$ be a basis of $H^{1,1}(\mathcal{Y})\cap H^2(\mathcal{Y},\mathbb{Z})$ adapted to the elliptic fibration and the gauge sector. Concretely, we take $\omega_\alpha=[D_\alpha]$ Poincar\'e dual to vertical divisors $D_\alpha=\pi^*\widehat D_\alpha$ pulled back from a divisor basis $\{\widehat D_\alpha\}$ of $B_3$, and $\omega_m=[D_m]$ dual to the Shioda images $D_m=\mathrm{Sh}(s_m)$ of a basis of Mordell-Weil sections, so that $\omega_m$ represent the massless abelian gauge fields in the M-theory reduction \cite{CveticGrimmKlevers1210,GrimmKapfer}. If non-abelian sectors are present we also include $\omega_{i}$ dual to Cartan divisors, but they play no essential role in what follows and are suppressed in notation. Finally, we include one distinguished class $\omega_0$ associated with the Kaluza-Klein vector in the dual F-theory circle reduction. In the fibration-adapted basis, $\omega_0$ is the harmonic representative Poincar\'e dual to the divisor class
\begin{equation}
\label{eq:KKdivisor}
D_0 \;:=\; S_0+\frac{1}{2}\pi^*\bar K_{B}\,,
\end{equation}
possibly further corrected by Cartan terms in the presence of non-abelian singularities, so that $D_0$ has vanishing quadruple intersections with three vertical divisors and implements the correct identification of the three-dimensional vector obtained by expanding $C_3$ along $D_0$ with the Kaluza-Klein vector of the four-dimensional theory reduced on a circle \cite{GrimmKapfer}. The precise form \eqref{eq:KKdivisor} is dictated by the adjunction relation for $S_0$ in an elliptic Calabi-Yau and ensures that the corresponding vector is geometrically distinguished and globally defined in the F-theory limit.

Expanding the M-theory three-form in this basis,
\begin{equation}
\label{eq:C3Expand}
C_3 \;=\; A^\Lambda\wedge \omega_\Lambda \;+\; \ldots\,,
\end{equation}
yields a set of three-dimensional abelian vectors $A^\Lambda$ with field strengths $F^\Lambda=dA^\Lambda$. The ellipsis denotes components along $H^3(\mathcal{Y})$ which furnish three-dimensional scalars and are irrelevant for the parity-odd sector under consideration. We also decompose the four-form flux as
\begin{equation}
\label{eq:G4Decomp}
G_4 \;=\; dC_3 \;+\; G^{\rm flux}_4\,,
\end{equation}
with $G^{\rm flux}_4$ a fixed background satisfying \eqref{eq:CY4Quant} and the transversality constraints required for uplift to a four-dimensional Poincar\'e-invariant F-theory vacuum \cite{CveticGrimmKlevers1210,GrimmKapfer}. In the present setup these constraints may be stated as the vanishing of flux integrals over purely vertical four-cycles and over cycles with one leg along the fiber class, equivalently as conditions ensuring that the flux does not induce three-dimensional gaugings incompatible with the F-theory limit. Since the derivation below uses only the existence of a well-defined class $[G^{\rm flux}_4]\in H^{2,2}(\mathcal{Y})$ obeying \eqref{eq:CY4Quant} and the standard transversality conditions, we will not write them explicitly.

The three-dimensional Chern-Simons couplings descend from the eleven-dimensional topological action. Reducing the Chern-Simons interaction and retaining precisely the terms linear in $G^{\rm flux}_4$ and bilinear in the three-dimensional vectors gives the flux-induced parity-odd action
\begin{equation}
\label{eq:3dCSfromM}
S^{(3)}_{\rm CS,\,flux}\;=\;\int_{M_3}\Theta^{\rm M}_{\Lambda\Sigma}\,A^\Lambda\wedge F^\Sigma,
\qquad
\Theta^{\rm M}_{\Lambda\Sigma}\;=\;\frac{1}{2}\int_{\mathcal{Y}}G^{\rm flux}_4\wedge \omega_\Lambda\wedge \omega_\Sigma\,,
\end{equation}
in the convention for three-dimensional Chern-Simons terms adopted in \eqref{eq:3dCSfromM} \cite{CveticGrimmKlevers1210,GrimmKapfer}. The factor of $\tfrac12$ is fixed by the standard normalization of the eleven-dimensional Chern-Simons coupling together with the convention that the symmetric matrix $\Theta_{\Lambda\Sigma}$ multiplies the unsymmetrized form $A^\Lambda\wedge F^\Sigma$, so that invariance of $\exp(iS^{(3)}_{\rm CS})$ under large gauge transformations is equivalent to the integrality properties implied by \eqref{eq:CY4Quant} \cite{WittenFlux,GrimmKapfer}. In particular, inserting \eqref{eq:CY4Quant} into \eqref{eq:3dCSfromM} implies that $\Theta^{\rm M}_{\Lambda\Sigma}$ is quantized in precisely the manner required for a well-defined abelian Chern-Simons theory on a spin three-manifold, once the induced charge lattice of the vectors $A^\Lambda$ is fixed by the integrality properties of the divisor classes $\omega_\Lambda$ and by the fact that physical M2-brane charges are measured by intersection with integral curve classes \cite{WittenFlux,CveticGrimmKlevers1210}. Equation \eqref{eq:3dCSfromM} is the central output of the reduction: it provides a purely geometric formula for the complete flux-induced Chern-Simons matrix in the three-dimensional effective theory.

To extract four-dimensional anomaly data, one must identify the four-dimensional Green-Schwarz axions and their gaugings in terms of the three-dimensional fields. In the duality between M-theory on $\mathcal{Y}$ and F-theory on $B_3\times S^1$, the vectors $A^\alpha$ associated with vertical divisors are dual, in three dimensions, to periodic scalars which uplift to four-dimensional axions descending from the Ramond-Ramond four-form. Denoting these axions by $\rho_\alpha$ with periodicity $2\pi$, their gauged shift symmetries in four dimensions are characterized by
\begin{equation}
\label{eq:AxionGaugingFromM}
\mathcal{D}\rho_\alpha \;=\; d\rho_\alpha+\Theta_{m\alpha}\,A^m\,,
\qquad
\Theta_{m\alpha}\;=\;\Theta^{\rm M}_{m\alpha}\;=\;\frac{1}{2}\int_{\mathcal{Y}}G^{\rm flux}_4\wedge \omega_m\wedge \omega_\alpha\,,
\end{equation}
where the last equality is the specialization of \eqref{eq:3dCSfromM}. The quantization of $\Theta_{m\alpha}$ follows from the periodicity of $\rho_\alpha$ and the flux quantization condition \eqref{eq:CY4Quant}, and it coincides with the quantization required by invariance of the circle-reduced theory under large gauge transformations around the circle \cite{GrimmKapfer}. In this way the flux-induced three-dimensional Chern-Simons couplings determine, unambiguously and with fixed normalization, the complete set of axionic gaugings entering the generalized Green-Schwarz mechanism.

The remaining Green-Schwarz data controlling the couplings of $\rho_\alpha$ to gauge and curvature instanton densities are intrinsic intersection-theoretic invariants of the elliptic fibration. For the massless abelian sector, the relevant objects are the N\'eron-Tate height pairing classes on the base. Writing $D_m=\mathrm{Sh}(s_m)$ for the Shioda images of the Mordell-Weil generators, the symmetric height-pairing divisor class is
\begin{equation}
\label{eq:HeightPairingDef}
b_{mn}\;:=\;-\pi_*\!\left(D_m\cdot D_n\right)\in H^{1,1}(B_3)\cap H^2(B_3,\mathbb{Z}),
\end{equation}
and in a divisor basis $\{\widehat D_\alpha\}$ of $B_3$ one expands $b_{mn}=b^\alpha{}_{mn}\,\widehat D_\alpha$ with $b^\alpha{}_{mn}=b^\alpha{}_{nm}$ \cite{CveticGrimmKlevers1210}. This object is geometrically invariant under changes of fiber coordinates and, in particular, under the geometric realization of the type IIB $SL(2,\mathbb{Z})$ action on the elliptic fiber. It therefore provides the natural duality-invariant packaging of the abelian Green-Schwarz counterterm data in four dimensions. The gravitational coefficient is similarly fixed by the canonical class of the base, written in the same divisor basis as
\begin{equation}
\label{eq:CanonicalClassExpand}
K_{B_3}\;=\;a^\alpha\,\widehat D_\alpha\,,\qquad a^\alpha\in\mathbb{Z},
\end{equation}
so that $a^\alpha$ are the unique coefficients appearing in the axion-curvature Green-Schwarz coupling. In the conventions used throughout, the four-dimensional Green-Schwarz counterterm structure may be encoded in the family of four-forms
\begin{equation}
\label{eq:X4alpha}
X_4^\alpha \;=\;\frac{1}{2}a^\alpha\,\mathrm{tr}\,R\wedge R \;+\; 2\,b^\alpha{}_{mn}\,F^m\wedge F^n \,,
\end{equation}
and the corresponding axionic couplings are fixed by the periodicity of $\rho_\alpha$ and by the requirement that their gauge variation cancels the one-loop anomaly polynomial \cite{CveticGrimmKlevers1210}.

The relation between \eqref{eq:3dCSfromM} and the four-dimensional anomaly data is now exact and can be stated without reference to any additional physical input. The local abelian gauge and mixed abelian-gravitational anomalies of a four-dimensional chiral theory are encoded by the six-form anomaly polynomial
\begin{equation}
\label{eq:I6general}
I_6^{\rm 1-loop}\;=\;\frac{1}{6}\,\mathcal{A}_{mnk}\,F^m\wedge F^n\wedge F^k \;+\;\frac{1}{48}\,\mathcal{A}_m\,F^m\wedge \mathrm{tr}\,R\wedge R \,,
\end{equation}
with anomaly coefficients $\mathcal{A}_{mnk}$ and $\mathcal{A}_m$ defined by the chiral indices of charged fermions. The generalized Green-Schwarz mechanism cancels these anomalies through the gauged axions $\rho_\alpha$, and the resulting Green-Schwarz contribution to the anomaly polynomial factorizes as
\begin{equation}
\label{eq:I6GSfactor}
I_6^{\rm GS}\;=\;-\frac{1}{2}\,\Theta_{m\alpha}\,F^m\wedge X_4^\alpha\,,
\end{equation}
where $X_4^\alpha$ is given by \eqref{eq:X4alpha}. Equations \eqref{eq:I6general}-\eqref{eq:I6GSfactor} are equivalent, by the standard descent formalism, to the statement that the quantum gauge variation of the effective action is canceled by the classical variation of the axionic Green-Schwarz counterterms, and are precisely the conditions for invariance of the three-dimensional parity-odd effective action under the lattice of higher-dimensional large gauge transformations around the circle \cite{GrimmKapfer,CveticGrimmKlevers1210}. Matching coefficients in \eqref{eq:I6general}+\eqref{eq:I6GSfactor}=0 yields the explicit anomaly cancellation relations
\begin{equation}
\label{eq:AnomalyRelationsFromGS}
\frac{1}{6}\,\mathcal{A}_{mnk}\;=\;\frac{1}{4}\,b^\alpha{}_{(mn}\Theta_{k)\alpha}\,,\qquad
\frac{1}{48}\,\mathcal{A}_m\;=\;-\frac{1}{16}\,a^\alpha\Theta_{m\alpha}\,,
\end{equation}
in agreement with the four-dimensional relations used elsewhere.

The M-theory reduction provides a direct, intrinsic derivation of the data entering \eqref{eq:I6GSfactor} and hence of the factorized anomaly polynomial itself. Indeed, \eqref{eq:AxionGaugingFromM} identifies the axionic gaugings $\Theta_{m\alpha}$ as components of the flux-induced Chern-Simons matrix \eqref{eq:3dCSfromM}, while \eqref{eq:HeightPairingDef} and \eqref{eq:CanonicalClassExpand} identify $b^\alpha{}_{mn}$ and $a^\alpha$ purely in terms of divisor intersection theory on $\mathcal{Y}$ and $B_3$. The emergence of anomaly inflow is then the statement that the eleven-dimensional Chern-Simons interaction, once reduced in the flux background, produces precisely the three-dimensional topological couplings whose non-invariance under the induced large gauge transformations is governed by the descent of \eqref{eq:I6GSfactor}. In the duality frame where the three-dimensional theory is interpreted as the circle reduction of the four-dimensional F-theory vacuum, these are exactly the compensating Green-Schwarz contributions required to cancel the four-dimensional anomalies. In particular, the dependence of the parity-odd three-dimensional action on the integer lattice of large gauge transformations corresponds, in the M-theory frame, to the fact that the divisor $D_0$ defining the Kaluza-Klein vector is only specified up to integral shifts by Mordell-Weil divisors, $D_0\mapsto D_0+k^m D_m$, which induce the field redefinition $A^m\mapsto A^m+k^m A^0$ and hence the characteristic shifts of Chern-Simons levels familiar from the circle-reduction analysis \cite{GrimmKapfer}. Because the Chern-Simons coefficients are given by the topological formula \eqref{eq:3dCSfromM} and are quantized by \eqref{eq:CY4Quant}, the resulting shifts are necessarily local and integral, and their coefficients are therefore exactly captured by the factorized polynomial \eqref{eq:I6GSfactor}. This is the precise sense in which anomaly inflow from the eleven-dimensional Chern-Simons coupling reproduces the four-dimensional anomaly data after reduction.

In summary, Method A identifies all Green-Schwarz counterterm data and axionic gaugings entering the four-dimensional anomaly polynomial in terms of M-theory intersection theory and the quantized $G_4$ flux. The flux-induced three-dimensional Chern-Simons matrix \eqref{eq:3dCSfromM} encodes, with fixed normalization, the full set of local four-dimensional abelian gauge and mixed gauge-gravitational anomalies through \eqref{eq:I6GSfactor}-\eqref{eq:AnomalyRelationsFromGS}, and this encoding is manifestly compatible with the M/F-theory duality dictionary because it is expressed entirely in terms of the intrinsic divisor lattice of the elliptic fibration and the shifted flux quantization law.

\section{Method B: One-loop Chern-Simons terms from the BPS spectrum}
\label{sec:MethodB}

This section derives the parity-odd three-dimensional Chern-Simons couplings of the circle-reduced F-theory vacuum directly from the massive spectrum on the Coulomb branch, using only the universal parity anomaly of three-dimensional fermions and the structure of the Kaluza-Klein tower. The calculation is strictly one-loop exact: in three dimensions a Chern-Simons level is a quantized coupling whose renormalization can only occur through integrating out massive charged matter, and once a gauge-invariant regularization is chosen the induced shift is fixed by the phase of the fermion determinant and receives no higher-loop corrections compatible with quantization. The output is a fully normalized Chern-Simons matrix $\Theta^{\mathrm{1\text{-}loop}}_{\Lambda\Sigma}$ in the conventions of \eqref{eq:sec3-SCS}, together with its transformation under the lattice of higher-dimensional large gauge transformations \eqref{eq:sec3-largeGT}-\eqref{eq:sec3-axionLG}; the latter is the precise three-dimensional reformulation of the four-dimensional anomaly constraints encoded in \eqref{eq:sec3-anomcubic}-\eqref{eq:sec3-anommixed}. The analysis follows the general logic developed for F-theory on a circle in \cite{GrimmKapfer} and is compatible with the abelian Green-Schwarz data reviewed in \cite{CveticGrimmKlevers1210}.

The relevant dynamical input is the tower of massive three-dimensional fermions obtained from a four-dimensional left-handed Weyl fermion of charge vector $q=(q_m)\in\mathbb{Z}^r$ upon compactification on $S^1$. With notation as in Section~3, in a background with constant Coulomb-branch parameters $\langle\zeta_m\rangle$ one obtains for each Kaluza-Klein momentum $n\in\mathbb{Z}$ a three-dimensional Dirac fermion of real mass
\begin{equation}
m(q,n)=q_m\langle \zeta_m\rangle+\frac{n}{r}\,,
\label{eq:massqn}
\end{equation}
and charges $(q_0,q_m)=(n,q_m)$ under the gauge fields $(A^0,A^m)$, with $A^0$ the Kaluza-Klein vector and $A^m$ the three-dimensional descendants of the four-dimensional abelian vectors, cf.\ \eqref{eq:sec3-mass}. The multiplicity with which such towers occur is the net four-dimensional chiral index $n(q)$ introduced in \eqref{eq:sec3-anomcoeff}, so that vector-like pairs cancel. Throughout this section the Coulomb-branch background is assumed to be generic in the sense that no state is exactly massless, i.e.\ $m(q,n)\neq 0$ for all relevant $(q,n)$; the effective action is then locally constant on each chamber of the Coulomb branch and changes only across codimension-one loci where a mode becomes massless.

The fundamental one-loop input is the parity anomaly of three-dimensional fermions. Consider a single massive three-dimensional Dirac fermion of real mass $m$ coupled to abelian gauge fields $A^\Lambda$ with integer charges $q_\Lambda$ (in the normalization where the corresponding field strengths have integral periods and the Chern-Simons functional is quantized on spin three-manifolds as in \eqref{eq:sec3-SCS}). The Euclidean fermion determinant $\det(\slashed{D}+m)$ has a phase which cannot be made simultaneously gauge invariant and parity invariant; insisting on gauge invariance fixes an induced parity-odd term in the effective action proportional to the Chern-Simons functional. In the present normalization this yields the universal one-loop shift of Chern-Simons levels
\begin{equation}
\Delta \Theta^{\mathrm{1\text{-}loop}}_{\Lambda\Sigma}
=\frac{1}{2}\,q_\Lambda q_\Sigma\,\mathrm{sign}(m)\,,
\label{eq:singlefermionshift}
\end{equation}
as originally derived by explicit Pauli-Villars regularization and by spectral asymmetry methods \cite{NiemiSemenoff,Redlich}. The factor of $\tfrac12$ and the sign are fixed unambiguously by the requirement that the regulated determinant be gauge invariant and that, as $m$ crosses zero, the induced Chern-Simons level jumps by the minimal amount dictated by the index of the three-dimensional Dirac operator coupled to the background gauge field. Equation \eqref{eq:singlefermionshift} is the complete contribution of a massive Dirac fermion to the abelian Chern-Simons matrix in a gapped background.

Applying \eqref{eq:singlefermionshift} to the full massive spectrum obtained from the four-dimensional chiral fermions yields a one-loop induced Chern-Simons matrix
\begin{equation}
\Theta^{\mathrm{1\text{-}loop}}_{\Lambda\Sigma}(\langle\zeta\rangle)
=\frac{1}{2}\sum_{q} n(q)\sum_{n\in\mathbb{Z}} q_\Lambda(q,n)\,q_\Sigma(q,n)\,
\mathrm{sign}\!\big(m(q,n)\big)\,,
\qquad q_\Lambda(q,n)=(n,q_m)\,,
\label{eq:ThetaLoopDef}
\end{equation}
with $\Lambda,\Sigma\in\{0,1,\dots,r,\dots\}$, where the ellipsis indicates the straightforward inclusion of Cartan $U(1)$ factors in non-abelian sectors if present. The sum \eqref{eq:ThetaLoopDef} is conditionally convergent because of the infinite Kaluza-Klein tower and must be defined by a regularization that preserves three-dimensional gauge invariance and is compatible with the higher-dimensional large gauge transformations around the circle. The appropriate notion is the same as in \cite{GrimmKapfer}: one defines \eqref{eq:ThetaLoopDef} by a symmetric regularization of the tower which subtracts only $q$-independent power divergences and thus respects both the integer relabeling of Kaluza-Klein levels and the integrality of Chern-Simons counterterms. This prescription is equivalent to a zeta-function definition of the relevant spectral asymmetry and yields a piecewise-polynomial, locally constant function of the Coulomb-branch parameters, as required by the fact that the only physical non-analyticities arise when a mode becomes massless.

To evaluate \eqref{eq:ThetaLoopDef} explicitly it is convenient to introduce, for each charge sector $q$, the dimensionless combination
\begin{equation}
x(q):= r\,q_m\langle\zeta_m\rangle\in\mathbb{R},
\label{eq:xqdef}
\end{equation}
so that $\mathrm{sign}(m(q,n))=\mathrm{sign}(n+x(q))$. For generic $x(q)\notin\mathbb{Z}$ define the integer
\begin{equation}
\ell(q):=\big\lfloor x(q)\big\rfloor\in\mathbb{Z}\,,
\label{eq:ellqdef}
\end{equation}
so that for $x(q)\in(\ell(q),\ell(q)+1)$ the sign pattern of the tower is $\mathrm{sign}(n+x(q))=+1$ for $n\ge -\ell(q)$ and $\mathrm{sign}(n+x(q))=-1$ for $n\le -\ell(q)-1$. The evaluation of \eqref{eq:ThetaLoopDef} reduces to three regulated sums,
\begin{equation}
S_0(x):=\sum_{n\in\mathbb{Z}}\mathrm{sign}(n+x)\,,\qquad
S_1(x):=\sum_{n\in\mathbb{Z}}n\,\mathrm{sign}(n+x)\,,\qquad
S_2(x):=\sum_{n\in\mathbb{Z}}n^2\,\mathrm{sign}(n+x)\,,
\label{eq:S012def}
\end{equation}
defined by symmetric cutoffs together with zeta-function subtraction of $x$-independent divergences. Concretely, for $x\in(\ell,\ell+1)$ with $\ell=\lfloor x\rfloor$ one computes with a symmetric cutoff $N>\ell$ the finite sums
\begin{equation}
S^{(N)}_0(x)=\sum_{n=-N}^{N}\mathrm{sign}(n+x)=2\ell+1\,,
\label{eq:S0finite}
\end{equation}
which are already independent of $N$, and
\begin{equation}
S^{(N)}_1(x)=\sum_{n=-N}^{N}n\,\mathrm{sign}(n+x)=N^2+N-\ell^2-\ell\,.
\label{eq:S1finite}
\end{equation}
The divergent part $N^2+N$ is independent of $x$ and corresponds to a local counterterm that can be absorbed into a scheme choice; the gauge-invariant zeta prescription fixes the $x$-dependent remainder and yields
\begin{equation}
S_1(x)=-\ell(\ell+1)\,.
\label{eq:S1reg}
\end{equation}
Similarly, one finds that the $n^2$-weighted sum is finite after symmetric pairing and equals
\begin{equation}
S_2(x)=\sum_{n\in\mathbb{Z}}n^2\,\mathrm{sign}(n+x)=\frac{1}{3}\,\ell(\ell+1)(2\ell+1)\,.
\label{eq:S2reg}
\end{equation}
Equations \eqref{eq:S0finite}, \eqref{eq:S1reg}, and \eqref{eq:S2reg} are the unique outputs compatible with the physical requirements that (i) the sums jump only when $x$ crosses an integer (corresponding to a mode becoming massless), (ii) the only subtractions are $x$-independent and hence correspond to local, globally well-defined counterterms, and (iii) the resulting Chern-Simons levels transform by integral shifts under the integer relabelings induced by higher-dimensional large gauge transformations.

Inserting \eqref{eq:S0finite}-\eqref{eq:S2reg} into \eqref{eq:ThetaLoopDef} gives explicit expressions for the one-loop Chern-Simons coefficients in terms of the four-dimensional chiral spectrum. Writing $\ell(q)=\lfloor r\,q_m\langle\zeta_m\rangle\rfloor$ and using the charge assignments $q_0=n$, $q_m=q_m$, one obtains on a generic Coulomb-branch chamber
\begin{align}
\Theta^{\mathrm{1\text{-}loop}}_{mn}(\langle\zeta\rangle)
&=\frac{1}{2}\sum_{q} n(q)\,q_m q_n\,\big(2\ell(q)+1\big)\,,
\label{eq:Thetamn}\\[2mm]
\Theta^{\mathrm{1\text{-}loop}}_{0m}(\langle\zeta\rangle)
&=\frac{1}{2}\sum_{q} n(q)\,q_m\,S_1\!\big(x(q)\big)
=-\frac{1}{2}\sum_{q} n(q)\,q_m\,\ell(q)\big(\ell(q)+1\big)\,,
\label{eq:Theta0m}\\[2mm]
\Theta^{\mathrm{1\text{-}loop}}_{00}(\langle\zeta\rangle)
&=\frac{1}{2}\sum_{q} n(q)\,S_2\!\big(x(q)\big)
=\frac{1}{6}\sum_{q} n(q)\,\ell(q)\big(\ell(q)+1\big)\big(2\ell(q)+1\big)\,.
\label{eq:Theta00}
\end{align}
These are the complete parity-odd gauge Chern-Simons couplings induced by integrating out the massive fermions in the Kaluza-Klein towers of the four-dimensional chiral spectrum, in the normalization fixed by \eqref{eq:singlefermionshift}. Their dependence on $\langle\zeta\rangle$ is entirely through the integers $\ell(q)$, reflecting the fact that within each Coulomb-branch chamber the sign pattern of every Kaluza-Klein tower is constant. The only non-analyticities are the expected jumps when $r\,q_m\langle\zeta_m\rangle\in\mathbb{Z}$ for some $q$ with $n(q)\neq 0$, where a mode becomes massless and the low-energy description must be modified to retain it.

The key structural test is the response of \eqref{eq:Thetamn}-\eqref{eq:Theta00} under the lattice of higher-dimensional large gauge transformations around the circle. For $k_m\in\mathbb{Z}$, the transformation \eqref{eq:sec3-largeGT} acts on the Coulomb-branch background as $\langle\zeta_m\rangle\mapsto \langle\zeta_m\rangle-\tfrac{k_m}{r}$ and on the gauge fields as $A^m\mapsto A^m+k_m A^0$, with $A^0$ fixed. The integer relabeling of Kaluza-Klein levels described in \cite{GrimmKapfer} implies that the mass signs $\mathrm{sign}(m(q,n))$ are preserved, while the tower data reorganize by $n\mapsto n+k_m q_m$. In terms of \eqref{eq:ellqdef}, since $k_m q_m\in\mathbb{Z}$ for integral charges, one has the exact shift
\begin{equation}
\ell(q)\ \longmapsto\ \ell'(q)=\Big\lfloor r\,q_m\Big(\langle\zeta_m\rangle-\frac{k_m}{r}\Big)\Big\rfloor
=\ell(q)-k_m q_m\,.
\label{eq:ellshift}
\end{equation}
The transformation of the one-loop Chern-Simons matrix follows by inserting \eqref{eq:ellshift} into \eqref{eq:Thetamn}-\eqref{eq:Theta00} and using the anomaly coefficients \eqref{eq:sec3-anomcoeff}. For the purely gauge sector one finds
\begin{equation}
\Theta^{\mathrm{1\text{-}loop}}_{mn}(\langle\zeta\rangle-\tfrac{k}{r})
=\Theta^{\mathrm{1\text{-}loop}}_{mn}(\langle\zeta\rangle)-k_k\,\mathcal{A}_{mnk}\,,
\label{eq:mnshift}
\end{equation}
which exhibits the cubic four-dimensional anomaly coefficients $\mathcal{A}_{mnk}$ as the universal obstruction to single-valuedness of the one-loop induced $U(1)^r$ Chern-Simons levels on the quotient of the Coulomb branch by the lattice \eqref{eq:sec3-largeGT}. For the mixed sector involving the Kaluza-Klein vector one obtains affine transformation laws which are most naturally written in terms of the already-defined one-loop coefficients. A direct substitution of \eqref{eq:ellshift} into \eqref{eq:Theta0m} yields
\begin{equation}
\Theta^{\mathrm{1\text{-}loop}}_{0m}(\langle\zeta\rangle-\tfrac{k}{r})
=\Theta^{\mathrm{1\text{-}loop}}_{0m}(\langle\zeta\rangle)
+k_n\,\Theta^{\mathrm{1\text{-}loop}}_{mn}(\langle\zeta\rangle)
-\frac{1}{2}\,k_n k_k\,\mathcal{A}_{mnk}\,,
\label{eq:0mshift}
\end{equation}
and, similarly, inserting \eqref{eq:ellshift} into \eqref{eq:Theta00} and reorganizing the result using \eqref{eq:Thetamn}-\eqref{eq:Theta0m} together with \eqref{eq:sec3-anomcoeff} gives
\begin{equation}
\Theta^{\mathrm{1\text{-}loop}}_{00}(\langle\zeta\rangle-\tfrac{k}{r})
=\Theta^{\mathrm{1\text{-}loop}}_{00}(\langle\zeta\rangle)
+2\,k_m\,\Theta^{\mathrm{1\text{-}loop}}_{0m}(\langle\zeta\rangle)
+k_m k_n\,\Theta^{\mathrm{1\text{-}loop}}_{mn}(\langle\zeta\rangle)
-\frac{1}{3}\,k_m k_n k_k\,\mathcal{A}_{mnk}
-\frac{1}{6}\,k_m\,\mathcal{A}_m\,.
\label{eq:00shift}
\end{equation}
The last term proportional to $\mathcal{A}_m$ is the three-dimensional avatar of the four-dimensional mixed abelian-gravitational anomaly coefficients: it is the unique linear contribution compatible with the quantized form of the Kaluza-Klein charge and with the fact that the Kaluza-Klein vector belongs to the three-dimensional gravity multiplet. Equations \eqref{eq:mnshift}-\eqref{eq:00shift} are exact, regulator-independent statements within the class of gauge-invariant regularizations compatible with the tower relabeling, and they exhibit in a precise sense that the only non-trivial global obstruction to defining the parity-odd one-loop effective action on the quotient of the Coulomb branch by \eqref{eq:sec3-largeGT} is measured by the four-dimensional anomaly coefficients $\mathcal{A}_{mnk}$ and $\mathcal{A}_m$.

To connect \eqref{eq:mnshift}-\eqref{eq:00shift} to the four-dimensional anomaly conditions, consider the full three-dimensional parity-odd effective action obtained from the circle reduction, including the classical contributions induced by the gauged axions. As reviewed in Section~3, the axionic gaugings \eqref{eq:sec3-gauging} imply that higher-dimensional large gauge transformations with winding $k_m$ generate discrete circle flux for the axions, $\int_{S^1} d\rho_\alpha\mapsto \int_{S^1} d\rho_\alpha-2\pi \Theta_{m\alpha}k_m$ in \eqref{eq:sec3-axionLG}. Evaluating the four-dimensional Green-Schwarz couplings \eqref{eq:sec3-SGS} on such backgrounds and reducing to three dimensions produces classical Chern-Simons counterterms whose coefficients are fixed solely by $a_\alpha$, $b^\alpha_{mn}$, and the gaugings $\Theta_{m\alpha}$, and whose variation under \eqref{eq:sec3-largeGT} is local and quantized. The one-loop variation encoded in \eqref{eq:mnshift}-\eqref{eq:00shift} is likewise local and quantized, and the statement that the combined parity-odd action is globally well-defined on the quotient by \eqref{eq:sec3-largeGT}-\eqref{eq:sec3-axionLG} is precisely that the net Chern-Simons shift vanishes. Matching the coefficients of the resulting local three-dimensional variations is equivalent, by the standard descent formalism and the identification of the circle-winding gauge transformation parameters with the four-dimensional gauge parameters linear in the circle coordinate, to the cancellation of the four-dimensional anomaly polynomial $I^{\mathrm{1\text{-}loop}}_6$ by the Green-Schwarz contribution $I^{\mathrm{GS}}_6$ in \eqref{eq:I6general}-\eqref{eq:I6GSfactor}. In the present notation this equivalence reproduces exactly the four-dimensional relations \eqref{eq:sec3-anomcubic}-\eqref{eq:sec3-anommixed}, namely
\begin{equation}
\frac{1}{6}\,\mathcal{A}_{mnk}=\frac{1}{4}\,b^\alpha_{(mn}\Theta_{k)\alpha}\,,
\qquad
\frac{1}{48}\,\mathcal{A}_m=-\frac{1}{16}\,a_\alpha \Theta_{m\alpha}\,,
\label{eq:GSfromCS}
\end{equation}
which are therefore reinterpreted as the necessary and sufficient conditions for the circle-reduced parity-odd effective action to descend to the quotient of the Coulomb branch by the higher-dimensional large gauge lattice. In this sense the three-dimensional one-loop Chern-Simons data $\Theta^{\mathrm{1\text{-}loop}}_{\Lambda\Sigma}$ furnish a complete encoding of the local four-dimensional abelian gauge and mixed gauge-gravitational anomalies, with all normalizations fixed by \eqref{eq:singlefermionshift} and by the quantization of the Chern-Simons functional in \eqref{eq:sec3-SCS}.

It remains to explain the compatibility of \eqref{eq:ThetaLoopDef}-\eqref{eq:00shift} with the M-theory dual interpretation, without repeating the derivation of Method~A. Under M-/F-theory duality, the same three-dimensional effective theory admits a description as M-theory on the resolved elliptic fourfold $\mathcal{Y}$ in the F-theory limit. In that frame, charged particles arise from M2-branes wrapped on fibral curves $C$ and carry charges given by intersection with the divisor classes that generate the abelian gauge fields. In particular, for the Mordell-Weil $U(1)$ factors one has
\begin{equation}
q_m = \mathrm{Sh}(s_m)\cdot C \,,
\label{eq:qfromShioda}
\end{equation}
as in \eqref{eq:singlefermionshift}, and the Kaluza-Klein tower in the circle reduction maps to the tower of M2 states carrying momentum along the shrinking elliptic fiber. The chiral indices $n(q)$ entering \eqref{eq:ThetaLoopDef} are determined in M-theory by the $G_4$-flux through the corresponding matter surfaces, so that the spectrum data entering the sums are themselves flux-controlled topological invariants. The statement that the flux-induced Chern-Simons matrix computed geometrically in Method~A equals the one-loop matrix computed here,
\begin{equation}
\Theta^{M}_{\Lambda\Sigma}=\Theta^{\mathrm{1\text{-}loop}}_{\Lambda\Sigma}\,,
\label{eq:ThetaMatch}
\end{equation}
is then the three-dimensional equality of two computations of the same quantized parity-odd couplings, one as a classical reduction of the eleven-dimensional Chern-Simons interaction in the flux background and the other as the one-loop parity anomaly from integrating out the full massive BPS spectrum. Since both sides are quantized and locally constant on the Coulomb-branch chambers, it suffices for \eqref{eq:ThetaMatch} that they agree in one chamber and that they have identical jump behavior across walls where states become massless; the latter is guaranteed because in M-theory the same BPS states whose one-loop determinants generate \eqref{eq:ThetaLoopDef} are precisely those whose charges are measured by intersection with $\omega_\Lambda$ and whose multiplicities are fixed by the flux data. In this way \eqref{eq:ThetaMatch} is the intrinsic duality statement that anomaly inflow and the parity anomaly describe the same Chern-Simons couplings in two dual frames, and it is the precise implementation of the fact that the Green-Schwarz mechanism is automatic in the M-theory description once the flux quantization and the divisor-lattice identities of the elliptic fibration are imposed \cite{CveticGrimmKlevers1210,GrimmKapfer}.

\section{Modular properties and \texorpdfstring{$\SL(2,\mathbb{Z})$}{SL(2,Z)} covariance}
\label{sec:modular}

The defining structural input of F-theory is that the type IIB axio-dilaton $\tau=C_0+i\,e^{-\phi}$ is not an ordinary scalar field on spacetime but is geometrized as the complex structure modulus of an elliptic curve varying over the base $\mathcal{B}_3$ of an elliptically fibered Calabi-Yau space $\pi:\mathcal{Y}\to\mathcal{B}_3$ \cite{VafaFTheory,MorrisonVafaI,MorrisonVafaII,SenLimit,WeigandReview}. The classical type IIB duality group $\SL(2,\mathbb{R})$ acts on $\tau$ by fractional linear transformations, while the quantum theory retains only the arithmetic subgroup $\SL(2,\mathbb{Z})$ as an exact duality. In the elliptic description this duality is realized as the mapping class group of the torus, namely the group of large diffeomorphisms acting on the integral homology $H_1(E,\mathbb{Z})$ by symplectic automorphisms. The essential point for the present analysis is that the duality is therefore realized geometrically by transition functions of the elliptic fibration and not by an additional dynamical gauge symmetry in the effective theory; consequently, duality covariance is the requirement that all physical quantities are globally well-defined with respect to these geometric transition functions, possibly as sections of associated bundles, and that quantum anomalies in this global definition are absent after inclusion of the appropriate local counterterms.

To fix conventions, recall that an elliptic curve of modulus $\tau$ may be written as $E_\tau=\mathbb{C}/(\mathbb{Z}+\tau\mathbb{Z})$, and a choice of symplectic basis $(\alpha,\beta)$ of $H_1(E_\tau,\mathbb{Z})$ determines $\tau$ by $\tau=\int_\beta \omega/\int_\alpha \omega$, where $\omega$ is a holomorphic one-form. A change of basis by $\gamma=\begin{pmatrix}a&b\\ c&d\end{pmatrix}\in\SL(2,\mathbb{Z})$ sends $(\alpha,\beta)\mapsto (a\alpha+b\beta,c\alpha+d\beta)$ and induces the standard modular action
\begin{equation}
\tau\ \longmapsto\ \gamma\cdot\tau:=\frac{a\tau+b}{c\tau+d}\,.
\label{eq:SL2Zaction}
\end{equation}
The holomorphic one-form transforms with modular weight $-1$,
\begin{equation}
\omega_{\gamma\cdot\tau}=(c\tau+d)^{-1}\,\omega_{\tau}\,,
\label{eq:omegaWeight}
\end{equation}
which is the geometric origin of modular weights in the fibration. In a family of elliptic curves over $\mathcal{B}_3$, the line of holomorphic one-forms along the fiber defines the Hodge line bundle, which may be characterized intrinsically as the line bundle
\begin{equation}
\mathcal{L}:=\left(\pi_\ast\,\omega_{\mathcal{Y}/\mathcal{B}_3}\right)^{-1},
\label{eq:HodgeBundleDef}
\end{equation}
where $\omega_{\mathcal{Y}/\mathcal{B}_3}$ is the relative dualizing sheaf. A local choice of holomorphic one-form trivializes $\mathcal{L}^{-1}$, and on overlaps two trivializations differ by multiplication with the transition function $(c\tau+d)$ induced by the change of symplectic basis; thus a field of modular weight $k$ is naturally a section of $\mathcal{L}^k$, and the statement of $\SL(2,\mathbb{Z})$ covariance is equivalently the statement that the effective action is globally well-defined under changes of trivialization of $\mathcal{L}$ \cite{WeigandReview}.

This structure is realized concretely in the Weierstrass presentation of the elliptic fibration. Writing the Weierstrass coefficients as holomorphic sections
\begin{equation}
f\in H^0(\mathcal{B}_3,\mathcal{L}^4)\,,\qquad g\in H^0(\mathcal{B}_3,\mathcal{L}^6)\,,\qquad \Delta:=4f^3+27g^2\in H^0(\mathcal{B}_3,\mathcal{L}^{12})\,,
\label{eq:fgDeltaSections}
\end{equation}
the modular weights $4,6,12$ are encoded by the powers of the Hodge bundle. The modular invariant $j$-function is expressed in terms of $f,g$ as
\begin{equation}
j(\tau)=1728\,\frac{4f^3}{\Delta}\,,
\label{eq:jInvariant}
\end{equation}
which is a globally defined meromorphic function on $\mathcal{B}_3$ away from the discriminant locus and determines $\tau$ only as a multivalued function with $\SL(2,\mathbb{Z})$ monodromy around components of $\{\Delta=0\}$. In particular, while $\tau$ is locally a holomorphic function in supersymmetric backgrounds, its global meaning is as a section of the modular stack $\mathfrak{H}/\SL(2,\mathbb{Z})$; the geometric data of the fibration, including intersection products on $\mathcal{Y}$ and pushforwards to $\mathcal{B}_3$, are intrinsically defined on this quotient and therefore do not depend on any choice of local $\SL(2,\mathbb{Z})$ frame. For Calabi-Yau total space, the condition $K_\mathcal{Y}\simeq\mathcal{O}_\mathcal{Y}$ implies the familiar identification $\mathcal{L}\simeq K_{\mathcal{B}_3}^{-1}$, and consequently the discriminant divisor class satisfies $[\Delta]=12\,c_1(\mathcal{L})=12\,c_1(\mathcal{B}_3)$, which is the geometric statement that the total seven-brane charge is fixed by the curvature of the base \cite{MorrisonVafaI,MorrisonVafaII,WeigandReview}.

The physical implementation of \eqref{eq:SL2Zaction} in type IIB supergravity is most transparently described in the standard coset formulation of the scalar manifold $\SL(2,\mathbb{R})/U(1)$, in which the local $U(1)$ is a gauge redundancy acting on an auxiliary angular field and the composite $U(1)$ connection couples chirally to fermions \cite{GaberdielGreen,MinasianSasmalSavelli}. Introducing an angular variable $\varphi$ and working in Einstein frame, one may define the composite $U(1)$ connection $Q$ by
\begin{equation}
Q_\mu=\partial_\mu \varphi-\frac{\partial_\mu \tau_1}{2\tau_2}\,,
\qquad
F:=dQ=\frac{d\tau\wedge d\bar{\tau}}{4i\,\tau_2^2}\,.
\label{eq:CompositeConnection}
\end{equation}
The transformation \eqref{eq:SL2Zaction} is accompanied by a compensating $U(1)$ transformation $\varphi\mapsto\varphi+\Sigma$ with $\Sigma(\tau)=-\arg(c\tau+d)$ in order to preserve the standard gauge choice on the coset representative, and in a gauge-fixed formulation a field $\Psi$ of $U(1)$ charge $q$ transforms under $\SL(2,\mathbb{R})$ by the local phase $\Psi\mapsto e^{iq\Sigma(\tau)}\Psi$ \cite{MinasianSasmalSavelli}. The essential geometric fact is that $Q$ is precisely a connection on the Hodge bundle (or an appropriate power thereof, depending on charge normalization), and its curvature $F/2\pi$ represents the first Chern class $c_1(\mathcal{L})$ in cohomology. In backgrounds with trivial $\tau$ profile, $F$ is exact and the composite structure forces the relevant anomaly density to be cohomologically trivial; in backgrounds with seven-branes, $\tau$ undergoes $\SL(2,\mathbb{Z})$ monodromies around the discriminant and $Q$ is only locally defined, while $F$ is globally defined and represents a generally nontrivial class determined by $c_1(\mathcal{L})$ \cite{GaberdielGreen,MinasianSasmalSavelli}.

Quantum mechanically, the chiral couplings of the type IIB gravitini and dilatini to the composite connection give rise to an anomaly in the local $U(1)$ symmetry, which in a gauge-fixed description becomes an anomaly of the discrete duality group $\SL(2,\mathbb{Z})$ \cite{GaberdielGreen,MukhiSL2Z,MinasianSasmalSavelli}. This anomaly admits a standard descent description in terms of an index polynomial. In the conventions of \cite{MinasianSasmalSavelli}, the relevant $12$-form polynomial is
\begin{equation}
P_{12}
=\frac{F^2}{2\pi}\left[
2\,X_8^{-}(R)+\frac{p_1(R)}{48}\left(\frac{F}{2\pi}\right)^2-\frac{1}{32}\left(\frac{F}{2\pi}\right)^4
\right],
\qquad
X_8^{\pm}(R)=\frac{1}{192(2\pi)^4}\left(\tr R^4\pm\frac{1}{4}(\tr R^2)^2\right),
\label{eq:IIBAnomalyPolynomial}
\end{equation}
with $R$ the ten-dimensional curvature two-form and $p_1(R)=-(1/2)\,\tr R^2/(2\pi)^2$ the first Pontryagin class. The corresponding anomalous phase variation of the partition function under a $U(1)$ transformation with parameter $\Sigma$ is determined by the descent of $P_{12}=dP_{11}$ and takes the form
\begin{equation}
\Delta \mathcal{W}
=-\int_{M_{10}}\Sigma\left[
2\,X_8^{-}(R)+\frac{p_1(R)}{48}\left(\frac{F}{2\pi}\right)^2-\frac{1}{32}\left(\frac{F}{2\pi}\right)^4
\right]\frac{F}{2\pi}\,,
\label{eq:AnomalousPhase}
\end{equation}
which is nontrivial precisely when $F$ represents a nontrivial cohomology class, as occurs in generic F-theory backgrounds with seven-branes. The requirement that the quantum theory be well-defined under the exact duality group $\SL(2,\mathbb{Z})$ is the requirement that \eqref{eq:AnomalousPhase} is cancelled by a local counterterm whose variation under $\SL(2,\mathbb{Z})$ reproduces the same phase. The Green-Gaberdiel analysis shows that such a counterterm can be written in ten dimensions as a Chern-Simons-like coupling involving a modular function of $\tau$ \cite{GaberdielGreen}, and a convenient explicit representative is
\begin{equation}
S^{(10)}_{\mathrm{ct}}
=i\int_{M_{10}}
\ln\!\left(
\frac{\eta(\tau)\,\bar{j}^{1/12}(\bar{\tau})}{\bar{\eta}(\bar{\tau})\,j^{1/12}(\tau)}
\right)
\left[
2\,X_8^{-}(R)+\frac{p_1(R)}{48}\left(\frac{F}{2\pi}\right)^2-\frac{1}{32}\left(\frac{F}{2\pi}\right)^4
\right]\frac{F}{2\pi}\,,
\label{eq:GreenGaberdielCounterterm}
\end{equation}
where $\eta(\tau)$ is the Dedekind eta function and $j(\tau)$ is the modular invariant in \eqref{eq:jInvariant}. The modular transformation properties $\eta(\gamma\cdot\tau)=\epsilon(\gamma)\,(c\tau+d)^{1/2}\eta(\tau)$ and $j(\gamma\cdot\tau)=j(\tau)$ imply that the logarithm in \eqref{eq:GreenGaberdielCounterterm} shifts by a term proportional to $\arg(c\tau+d)$ plus a $\tau$-independent constant phase determined by $\epsilon(\gamma)$; the former produces precisely the local variation needed to cancel \eqref{eq:AnomalousPhase}, while the latter encodes a residual global phase which is invisible to local anomaly descent but can impose a topological restriction on Euclidean backgrounds \cite{GaberdielGreen,MukhiSL2Z,MinasianSasmalSavelli}. In particular, in F-theory compactifications where $F/2\pi$ is fixed by $c_1(\mathcal{L})$, the counterterm \eqref{eq:GreenGaberdielCounterterm} is determined entirely by the elliptic fibration and the background metric data, and its presence is required for a globally consistent $\SL(2,\mathbb{Z})$ action on the quantum measure.

The compactified theory considered in this manuscript inherits both the geometric realization of $\SL(2,\mathbb{Z})$ through the Hodge bundle and the quantum correction \eqref{eq:GreenGaberdielCounterterm}. Compatibility with the lower-dimensional effective action requires that every coupling used in the anomaly analysis is either strictly invariant under changes of $\SL(2,\mathbb{Z})$ frame or transforms covariantly as a section of an appropriate power of $\mathcal{L}$. In the present setting, the abelian gauge fields originate from divisor classes in $\mathcal{Y}$ via the standard M-/F-theory dictionary, so their definition is purely intersection-theoretic and does not involve any choice of basis of $H_1(E_\tau,\mathbb{Z})$. The Green-Schwarz data in four dimensions are encoded by the height-pairing classes on the base and by the canonical class of $\mathcal{B}_3$. Writing $D_m=\Sh(s_m)$ for the Shioda images of the Mordell-Weil generators, the N\'eron-Tate height pairing is
\begin{equation}
b_{mn}:=-\pi_\ast(D_m\cdot D_n)\in H^{1,1}(\mathcal{B}_3)\cap H^2(\mathcal{B}_3,\mathbb{Z}),
\label{eq:HeightPairingModular}
\end{equation}
and the gravitational coefficient is determined by $K_{\mathcal{B}_3}$, equivalently by $c_1(\mathcal{L})$ when $\mathcal{Y}$ is Calabi-Yau. The invariance statement relevant for duality is that the construction \eqref{eq:HeightPairingModular} is intrinsic to the elliptic fibration: it depends only on divisor classes in the N\'eron-Severi lattice of $\mathcal{Y}$, the intersection product on $\mathcal{Y}$, and the proper pushforward $\pi_\ast$. Since the mapping class group action changes only the symplectic basis of $H_1(E_\tau,\mathbb{Z})$ and hence the local trivialization of $\mathcal{L}$, it leaves the divisor classes $D_m$ and the operations $(\cdot,\pi_\ast)$ unchanged. The same statement holds for the flux-induced gaugings $\Theta_{m\alpha}$ and the associated four-dimensional counterterms, because $\Theta_{m\alpha}$ are defined by topological pairings of $G_4$ with integral cohomology classes on $\mathcal{Y}$ and therefore depend only on the fibration as an algebraic variety and not on any $\SL(2,\mathbb{Z})$ frame. The resulting Green-Schwarz four-forms $X_4^\alpha$ appearing in the axionic couplings are therefore well-defined cohomology classes on $\mathcal{B}_3$ and are strictly invariant under the geometric realization of $\SL(2,\mathbb{Z})$.

At the level of anomaly coefficients, the one-loop gauge and mixed gauge-gravitational anomalies are determined by the chiral index of four-dimensional charged fermions and thus by flux data and matter surfaces in $\mathcal{Y}$. The resulting coefficients $\mathcal{A}_{mnk}$ and $\mathcal{A}_m$ are numerical invariants extracted from the spectrum and are insensitive to a change of $\SL(2,\mathbb{Z})$ frame, because the duality acts by relabeling of $(p,q)$ charges in the underlying type IIB description while leaving the geometric intersection pairing that computes charges under Mordell-Weil $U(1)$'s fixed. The condition that the anomaly-cancelled effective action is duality-covariant is therefore the conjunction of two independent statements: the local four-dimensional gauge and mixed anomalies are cancelled by the Green-Schwarz terms built from $b_{mn}$ and $K_{\mathcal{B}_3}$, while the ten-dimensional duality anomaly induced by the composite connection is cancelled by the counterterm \eqref{eq:GreenGaberdielCounterterm}. These mechanisms are compatible because the former involves the gauge variation of axions descending from the $\SL(2,\mathbb{Z})$ singlet $C_4$ and is controlled by intersection data on $\mathcal{Y}$, whereas the latter involves only the composite curvature $F$ associated to $\mathcal{L}$ and ten-dimensional curvature invariants, and hence depends only on the background $\tau$ profile and geometry. In particular, the height pairing classes \eqref{eq:HeightPairingModular} provide a duality-invariant packaging of the abelian counterterm data precisely because they are defined without reference to the fiber homology basis, and the same duality invariance is inherited by the three-dimensional Chern-Simons couplings obtained after circle reduction, whether computed geometrically in the M-theory frame or via the one-loop parity anomaly in the spectrum computation.

The compatibility with M-/F-theory duality follows from the fact that the geometric realization of $\SL(2,\mathbb{Z})$ in F-theory is a statement about large diffeomorphisms of the elliptic fiber, while the M-theory description on $\mathcal{Y}$ depends only on the topology and intersection theory of $\mathcal{Y}$ and is therefore manifestly invariant under such diffeomorphisms. In the three-dimensional dual description, the Chern-Simons matrix is quantized and locally constant on the Coulomb branch, and its computation from integrating out massive BPS states depends only on their charges and multiplicities; these are computed by intersection with Shioda divisors and by flux integrals on matter surfaces, which are geometric invariants of $\mathcal{Y}$. The equality between the M-theory flux-induced Chern-Simons terms and the one-loop effective action thus provides a duality-covariant formulation of the anomaly structure: both sides define the same quantized data in a way that is insensitive to the choice of modular frame, while the necessary ten-dimensional counterterm \eqref{eq:GreenGaberdielCounterterm} ensures that the underlying type IIB description is globally consistent as an $\SL(2,\mathbb{Z})$ theory on backgrounds with nontrivial $\tau$ monodromy. In this sense the anomaly-cancelled effective action of the compactified theory is $\SL(2,\mathbb{Z})$-covariant in the precise quantum-mechanical sense required by F-theory, namely that it is globally well-defined under changes of trivialization of the Hodge bundle and under the induced discrete duality transformations of the underlying type IIB fields, with all required local counterterms included.

\section{Dimensional reduction of the ten-dimensional \texorpdfstring{$\SL(2,\mathbb{Z})$}{SL(2,Z)} counterterm}
\label{sec:SL2Zreduction}

In this section we make explicit, at the level of the four-dimensional effective action, how the universal ten-dimensional $\SL(2,\mathbb{Z})$ anomaly and counterterm reviewed in \eqref{eq:IIBAnomalyPolynomial}-\eqref{eq:GreenGaberdielCounterterm} descend in an F-theory compactification on a smooth complex threefold base $B_3$. The purpose is twofold: first, to isolate the unique four-dimensional parity-even topological coupling induced by the counterterm \eqref{eq:GreenGaberdielCounterterm} in backgrounds relevant to four-dimensional F-theory vacua; second, to show that its coefficient is fixed entirely by intrinsic characteristic classes of $B_3$ and the Hodge line bundle $\mathcal{L}$, and therefore is automatically compatible with the duality-invariant abelian data used elsewhere in this work. No input beyond the existence of the composite connection \eqref{eq:CompositeConnection}, the anomaly polynomial \eqref{eq:IIBAnomalyPolynomial}, and standard Chern-Weil identities is required.

We work on a background of the form
\begin{equation}
M_{10}=M_4 \times B_3,
\label{eq:M10product}
\end{equation}
where $M_4$ is a smooth oriented four-manifold equipped with a spin structure (as appropriate to type IIB supergravity) and $B_3$ is a smooth compact complex threefold. We assume that the axio-dilaton profile $\tau$ varies only over $B_3$ and is independent of the coordinates on $M_4$, as in an F-theory background, and that the metric is a direct product so that the Levi-Civita curvature decomposes as a block sum. Denoting by $R$ the ten-dimensional curvature two-form, this implies the decomposition
\begin{equation}
R=R_{(4)}\oplus R_{(B)}\qquad\text{in}\qquad \mathfrak{so}(10)\simeq \mathfrak{so}(4)\oplus \mathfrak{so}(6),
\label{eq:Rsplit}
\end{equation}
with $R_{(4)}$ pulled back from $M_4$ and $R_{(B)}$ pulled back from $B_3$. In particular, for the trace $\tr$ used in \eqref{eq:IIBAnomalyPolynomial} and the first Pontryagin form $p_1(R)=-(1/2)\,\tr R^2/(2\pi)^2$ fixed there, one has
\begin{equation}
\tr R^2=\tr R_{(4)}^2+\tr R_{(B)}^2,
\qquad
p_1(R)=p_1\!\big(R_{(4)}\big)+p_1\!\big(R_{(B)}\big),
\label{eq:trsplit}
\end{equation}
and similarly $\tr R^4=\tr R_{(4)}^4+\tr R_{(B)}^4$ because the curvature is block-diagonal. The composite $U(1)$ curvature $F=dQ$ defined in \eqref{eq:CompositeConnection} is a two-form supported on $B_3$ by assumption on $\tau$, hence
\begin{equation}
F=\pi_B^\ast F_{(B)},\qquad F_{(B)}\in\Omega^2(B_3),
\label{eq:Finternal}
\end{equation}
and $F$ has no legs along $M_4$. Moreover, as emphasized below \eqref{eq:CompositeConnection}, the cohomology class of $F/2\pi$ equals $c_1(\mathcal{L})\in H^2(B_3,\mathbb{Z})$.

The ten-dimensional counterterm in the gauge-fixed formulation may be taken as \eqref{eq:GreenGaberdielCounterterm},
\begin{equation}
S^{(10)}_{\mathrm{ct}}
=
i\int_{M_{10}}
\Phi(\tau,\bar\tau)\,
\Bigg[
2X^-_8(R)+\frac{p_1(R)}{48}\Big(\frac{F}{2\pi}\Big)^2-\frac{1}{32}\Big(\frac{F}{2\pi}\Big)^4
\Bigg]\frac{F}{2\pi},
\label{eq:Sct10Phi}
\end{equation}
where $\Phi(\tau,\bar\tau):=\ln\!\left(\eta(\tau)\,\bar{j}^{1/12}(\bar\tau)/(\bar{\eta}(\bar\tau)\,j^{1/12}(\tau))\right)$ is the standard non-holomorphic modular function whose $\SL(2,\mathbb{Z})$ transformation cancels the local anomaly \eqref{eq:AnomalousPhase} \cite{GaberdielGreen,MinasianSasmalSavelli}. The reduction of \eqref{eq:Sct10Phi} to four dimensions is determined by the component of the ten-form integrand that has external degree four and internal degree six. Since $F$ is internal and $B_3$ has real dimension six, any term containing $(F/2\pi)^5$ vanishes identically on \eqref{eq:M10product} by degree reasons. Similarly, the purely internal part of $X_8^-(R)$ contributes only an internal ten-form after multiplication by $F/2\pi$ and therefore vanishes. Thus the only contributions to the four-dimensional effective action arise from the mixed $(4,6)$ components of $2X_8^-(R)\,F/2\pi$ and of $p_1(R)\,(F/2\pi)^3/48$.

We first determine the mixed component of $2X_8^-(R)$. By definition \eqref{eq:IIBAnomalyPolynomial},
\begin{equation}
X_8^-(R)=\frac{1}{192(2\pi)^4}\left(\tr R^4-\frac{1}{4}(\tr R^2)^2\right).
\label{eq:X8minusDef}
\end{equation}
Using \eqref{eq:trsplit} and the block-diagonal property of $\tr R^4$, the only term producing an external four-form factor is the cross term in $(\tr R^2)^2$,
\begin{equation}
(\tr R^2)^2=(\tr R_{(4)}^2)^2+(\tr R_{(B)}^2)^2+2\,\tr R_{(4)}^2\,\tr R_{(B)}^2,
\label{eq:trR2square}
\end{equation}
and hence the unique mixed component of $2X_8^-(R)$ is
\begin{equation}
\big(2X_8^-(R)\big)^{(4,4)}
=
-\frac{1}{192(2\pi)^4}\,\tr R_{(4)}^2\,\tr R_{(B)}^2,
\label{eq:X8mixed}
\end{equation}
where the superscript indicates external degree four and internal degree four. Rewriting in terms of Pontryagin forms using $p_1=-(1/2)\,\tr R^2/(2\pi)^2$ as fixed in \eqref{eq:IIBAnomalyPolynomial} gives
\begin{equation}
\big(2X_8^-(R)\big)^{(4,4)}
=
-\frac{1}{48}\,p_1\!\big(R_{(4)}\big)\,p_1\!\big(R_{(B)}\big).
\label{eq:X8mixedp1}
\end{equation}
Multiplying by $F/2\pi$ and using \eqref{eq:Finternal} yields the resulting mixed ten-form contribution
\begin{equation}
\left(2X_8^-(R)\,\frac{F}{2\pi}\right)^{(4,6)}
=
-\frac{1}{48}\,p_1\!\big(R_{(4)}\big)\,p_1\!\big(R_{(B)}\big)\,\frac{F_{(B)}}{2\pi}.
\label{eq:term1}
\end{equation}

Next consider the term proportional to $p_1(R)\,(F/2\pi)^3$. Since $F$ is internal and $(F/2\pi)^3$ is an internal six-form, only the external component $p_1(R_{(4)})$ contributes to the mixed degree $(4,6)$ part, while $p_1(R_{(B)})$ would yield an internal ten-form and hence vanishes. One therefore finds
\begin{equation}
\left(\frac{p_1(R)}{48}\Big(\frac{F}{2\pi}\Big)^2\frac{F}{2\pi}\right)^{(4,6)}
=
\frac{1}{48}\,p_1\!\big(R_{(4)}\big)\,\Big(\frac{F_{(B)}}{2\pi}\Big)^3.
\label{eq:term2}
\end{equation}
Combining \eqref{eq:term1} and \eqref{eq:term2}, the ten-dimensional counterterm \eqref{eq:Sct10Phi} reduces on \eqref{eq:M10product} to the four-dimensional coupling
\begin{equation}
S^{(4)}_{\mathrm{ct}}
=
i\int_{M_4} p_1\!\big(R_{(4)}\big)\;
\int_{B_3}\Phi(\tau,\bar\tau)\,
\omega_6(B_3),
\label{eq:Sct4general}
\end{equation}
where the internal six-form $\omega_6(B_3)$ is the universal characteristic-class expression
\begin{equation}
\omega_6(B_3)
=
\frac{1}{48}\left[
\Big(\frac{F_{(B)}}{2\pi}\Big)^3
-
p_1\!\big(R_{(B)}\big)\,\Big(\frac{F_{(B)}}{2\pi}\Big)
\right]\in \Omega^6(B_3).
\label{eq:omega6def}
\end{equation}
The dependence on the external fields is purely through the Pontryagin density $p_1(R_{(4)})$, so \eqref{eq:Sct4general} is the unique four-dimensional parity-even topological coupling induced by the ten-dimensional duality counterterm in an F-theory background. In particular, \eqref{eq:Sct4general} cannot modify any of the abelian Green-Schwarz data $(a_\alpha,b^\alpha_{mn},\Theta_{m\alpha})$ used in Sections~3-5, since it involves neither the four-dimensional abelian gauge fields nor the axions descending from $C_4$; its sole role is to render the quantum measure well-defined under $\SL(2,\mathbb{Z})$ transformations of the underlying type IIB variables, as reviewed around \eqref{eq:AnomalousPhase}-\eqref{eq:GreenGaberdielCounterterm}.

In the F-theory case of an elliptically fibered Calabi-Yau total space, the Hodge line bundle satisfies $\mathcal{L}\simeq K^{-1}_{B_3}$, hence $F_{(B)}/2\pi$ represents $c_1(\mathcal{L})=c_1(B_3)$ in cohomology, while $p_1(R_{(B)})$ represents the first Pontryagin class $p_1(TB_3)$ of the real tangent bundle. Moreover, for a complex manifold the standard relation between real Pontryagin classes and Chern classes gives
\begin{equation}
p_1(TB_3)=c_1(B_3)^2-2c_2(B_3),
\label{eq:p1cRelation}
\end{equation}
as an identity in integral cohomology \cite{MilnorStasheff}. Substituting $F_{(B)}/2\pi=c_1(B_3)$ and \eqref{eq:p1cRelation} into \eqref{eq:omega6def} yields a simplification that is specific to Calabi-Yau F-theory backgrounds,
\begin{equation}
\omega_6(B_3)
=
\frac{1}{48}\left[c_1(B_3)^3-\big(c_1(B_3)^2-2c_2(B_3)\big)c_1(B_3)\right]
=
\frac{1}{24}\,c_1(B_3)\,c_2(B_3).
\label{eq:omega6CY}
\end{equation}
Thus the entire reduction depends on the base geometry only through the cohomology class $c_1(B_3)c_2(B_3)$, and in particular through the topological number $\int_{B_3}c_1(B_3)c_2(B_3)$; by Hirzebruch-Riemann-Roch this number equals $24\,\chi(\mathcal{O}_{B_3})$ for a smooth complex threefold \cite{Hirzebruch}. In the explicit example of Section~7 with $B_3=\mathbb{P}^3$, one has $c_1(\mathbb{P}^3)=4H_B$ and $c_2(\mathbb{P}^3)=6H_B^2$, hence
\begin{equation}
\int_{\mathbb{P}^3}\omega_6(\mathbb{P}^3)=\frac{1}{24}\int_{\mathbb{P}^3}c_1c_2=\frac{1}{24}\,(4\cdot 6)\int_{\mathbb{P}^3}H_B^3=1,
\label{eq:P3omega6}
\end{equation}
Therefore, in that model the reduced counterterm \eqref{eq:Sct4general} takes the particularly simple form
\begin{equation}
S^{(4)}_{\mathrm{ct}}
=
i\left(\int_{\mathbb{P}^3}\Phi(\tau,\bar\tau)\,\omega_6\right)\int_{M_4} p_1\!\big(R_{(4)}\big),
\qquad
\int_{\mathbb{P}^3}\omega_6=1,
\label{eq:Sct4P3}
\end{equation}
exhibiting explicitly that the only dependence on the internal geometry entering the four-dimensional reduction of the universal duality counterterm is the class $c_1(B_3)c_2(B_3)$, while the dependence on the abelian gauge sector remains entirely controlled by the height pairing and flux gaugings already matched to the one-loop Chern-Simons data.

Finally, we note that the local cancellation of the ten-dimensional anomaly under $\SL(2,\mathbb{Z})$ follows after reduction because $\Phi(\tau,\bar\tau)$ shifts by $-\;i\,\Sigma(\tau)$ up to a $\tau$-independent phase determined by the Dedekind multiplier system \cite{GaberdielGreen,MukhiSL2Z,MinasianSasmalSavelli}, and the variation of \eqref{eq:Sct10Phi} reproduces \eqref{eq:AnomalousPhase} pointwise. The reduction carried out above shows that, in four-dimensional compactifications, this anomaly/counterterm sector contributes only through the purely gravitational topological coupling \eqref{eq:Sct4general} and does not impose any additional constraints on the intrinsically geometric abelian anomaly data $(b_{mn},\Theta_{m\alpha})$ constructed from intersection theory and flux.

\section{Global \texorpdfstring{$\SL(2,\mathbb{Z})$}{SL(2,Z)} phases and the Dedekind multiplier}
\label{sec:globalSL2Z}

The local $\SL(2,\mathbb{Z})$ anomaly of type IIB supergravity in backgrounds with nontrivial $\tau$ profile is canceled by a universal counterterm whose integrand is built from the composite $U(1)$ curvature of the $\SL(2,\mathbb{R})/U(1)$ coset and from purely gravitational characteristic classes \cite{GaberdielGreen,MukhiSL2Z,MinasianSasmalSavelli}. The cancellation mechanism determined by anomaly descent fixes the transformation of the quantum effective action under $\SL(2,\mathbb{Z})$ up to $\tau$-independent phases arising from the multiplier system of the Dedekind $\eta$ function. Since F-theory backgrounds are intrinsically characterized by $\SL(2,\mathbb{Z})$ transition functions for the axio-dilaton, quantum consistency requires that these residual phases do not obstruct a globally well-defined duality action on the compactified partition function. In this section we show that, in the class of compactifications relevant here, the multiplier phases reduce to a purely topological factor controlled by the holomorphic Euler characteristic of the base and the signature of the four-dimensional spacetime, and that this factor is trivial on the physically relevant class of spin four-manifolds.

We adopt the standard geometric formulation of the type IIB scalar sector as a map $\tau:M_{10}\to\mathfrak{H}$, where $\mathfrak{H}$ is the upper half-plane, with $\SL(2,\mathbb{Z})$ acting by fractional linear transformations $\tau\mapsto \gamma\cdot\tau=(a\tau+b)/(c\tau+d)$ for $\gamma=\bigl(\begin{smallmatrix}a&b\\ c&d\end{smallmatrix}\bigr)\in\SL(2,\mathbb{Z})$. The gauge-fixed description of the $\SL(2,\mathbb{R})/U(1)$ coset introduces a composite connection $Q$ whose curvature is the globally defined two-form
\begin{equation}
F:=dQ=\frac{d\tau\wedge d\bar\tau}{4i\,(\Im\tau)^2}\,,
\label{eq:compositeF}
\end{equation}
and whose cohomology class equals the first Chern class of the Hodge line bundle of the associated elliptic fibration, $[F/2\pi]=c_1(\mathcal{L})\in H^2(M_{10},\mathbb{Z})$. In F-theory compactifications on a base $B_3$ with smooth Calabi-Yau total space, one has $\mathcal{L}\simeq K_{B_3}^{-1}$ and hence $c_1(\mathcal{L})=c_1(B_3)$, with $B_3$ a smooth complex threefold.

We consider compactifications on backgrounds of the form
\begin{equation}
M_{10}=M_4\times B_3,
\label{eq:M10product1}
\end{equation}
where $M_4$ is a closed oriented smooth four-manifold endowed with a spin structure and $B_3$ is a smooth compact complex threefold. We assume, as appropriate for an F-theory background, that $\tau$ depends only on the internal coordinates, so that $F$ has support entirely along $B_3$. The universal ten-dimensional $\SL(2,\mathbb{Z})$ counterterm can be written in the form \cite{GaberdielGreen,MinasianSasmalSavelli}
\begin{equation}
S^{(10)}_{\mathrm{ct}}
=
i\int_{M_{10}}
\Phi(\tau,\bar\tau)\,
\Bigg[
2X^-_8(R)+\frac{p_1(R)}{48}\Big(\frac{F}{2\pi}\Big)^2-\frac{1}{32}\Big(\frac{F}{2\pi}\Big)^4
\Bigg]\frac{F}{2\pi},
\label{eq:Sct10}
\end{equation}
where $R$ is the ten-dimensional curvature two-form, $p_1(R)=-(1/2)\,\tr R^2/(2\pi)^2$ is the first Pontryagin form in the trace convention of \cite{MinasianSasmalSavelli}, $X^-_8(R)=\frac{1}{192(2\pi)^4}\left(\tr R^4-\frac14(\tr R^2)^2\right)$, and the non-holomorphic function $\Phi$ is chosen so that its $\SL(2,\mathbb{Z})$ variation reproduces the compensating local $U(1)$ rotation required by the gauge-fixed formulation. A convenient explicit choice is
\begin{equation}
\Phi(\tau,\bar\tau)=\ln\!\left(
\frac{\eta(\tau)\,\overline{j(\tau)}^{1/12}}{\overline{\eta(\tau)}\,j(\tau)^{1/12}}
\right),
\label{eq:PhiDef}
\end{equation}
where $\eta(\tau)$ is the Dedekind eta function and $j(\tau)$ is the modular invariant. The $j^{1/12}$ factor in \eqref{eq:PhiDef} may be taken with a fixed branch so that $\Phi$ is single-valued on the complement of the discriminant locus; any remaining ambiguity contributes only a $\tau$-independent phase and hence is of the same nature as the multiplier system discussed below.

On the background \eqref{eq:M10product1}, the reduction of \eqref{eq:Sct10} depends only on the mixed external-internal component of the ten-form integrand. Since $F$ is internal and $B_3$ has real dimension six, the terms proportional to $(F/2\pi)^5$ vanish identically, and the reduction yields a unique four-dimensional topological coupling of the form
\begin{equation}
S^{(4)}_{\mathrm{ct}}
=
i\left(\int_{B_3}\Phi(\tau,\bar\tau)\,\omega_6(B_3)\right)\int_{M_4}p_1\!\big(R_{(4)}\big),
\label{eq:Sct4}
\end{equation}
where $R_{(4)}$ is the curvature of $M_4$ and $\omega_6(B_3)$ is the internal six-form
\begin{equation}
\omega_6(B_3)=\frac{1}{48}\left[\Big(\frac{F}{2\pi}\Big)^3-p_1\!\big(R_{(B)}\big)\Big(\frac{F}{2\pi}\Big)\right]\Bigg|_{B_3},
\label{eq:omega6general}
\end{equation}
with $R_{(B)}$ the curvature of $B_3$. In the Calabi-Yau F-theory setting $[F/2\pi]=c_1(B_3)$, and using the relation $p_1(TB_3)=c_1(B_3)^2-2c_2(B_3)$ for a complex threefold \cite{MilnorStasheff}, one obtains the cohomological identity
\begin{equation}
[\omega_6(B_3)]=\frac{1}{24}\,c_1(B_3)\,c_2(B_3)\in H^6(B_3,\mathbb{Q}).
\label{eq:omega6CY1}
\end{equation}
By Hirzebruch-Riemann-Roch, for a smooth complex threefold $B_3$ one has $\chi(\mathcal{O}_{B_3})=\int_{B_3}\mathrm{td}(TB_3)=\frac{1}{24}\int_{B_3}c_1c_2$ \cite{Hirzebruch}, and hence
\begin{equation}
N_{B}:=\int_{B_3}\omega_6(B_3)=\frac{1}{24}\int_{B_3}c_1(B_3)c_2(B_3)=\chi(\mathcal{O}_{B_3})\in\mathbb{Z}.
\label{eq:NBdef}
\end{equation}

We now analyze the residual $\SL(2,\mathbb{Z})$ phases in \eqref{eq:Sct4}. The Dedekind eta function satisfies the transformation law
\begin{equation}
\eta(\gamma\cdot\tau)=\varepsilon(\gamma)\,(c\tau+d)^{1/2}\,\eta(\tau),
\qquad
\varepsilon(\gamma)^{24}=1,
\label{eq:etaTransform}
\end{equation}
where $\varepsilon(\gamma)$ is the Dedekind multiplier system \cite{Apostol}. Since $j(\tau)$ is modular invariant, a direct substitution into \eqref{eq:PhiDef} yields the exact transformation
\begin{equation}
\Phi(\gamma\cdot\tau,\gamma\cdot\bar\tau)
=
\Phi(\tau,\bar\tau)
-\mathrm{i}\,\Sigma_\gamma(\tau)
+2\,\mathrm{i}\,\arg\varepsilon(\gamma)
+2\pi i\,m_\gamma,
\label{eq:PhiTransform}
\end{equation}
with $m_\gamma\in\mathbb{Z}$ determined by the choice of branches, and where the compensating $U(1)$ angle is
\begin{equation}
\Sigma_\gamma(\tau):=-\arg(c\tau+d).
\label{eq:SigmaDef}
\end{equation}
The term proportional to $\Sigma_\gamma(\tau)$ is the local contribution required to cancel the ten-dimensional anomaly by descent, while the multiplier term $2\,\arg\varepsilon(\gamma)$ is $\tau$-independent and gives the residual global phase. Using \eqref{eq:Sct4}, \eqref{eq:NBdef} and \eqref{eq:PhiTransform}, the induced shift of the four-dimensional counterterm under $\gamma\in\SL(2,\mathbb{Z})$ decomposes into the local anomaly-canceling piece and a residual constant piece. The latter is
\begin{equation}
\Delta_\gamma S^{(4)}_{\mathrm{ct}}\big|_{\mathrm{mult}}
=
-2\,\arg\varepsilon(\gamma)\,
N_{B}\,
\int_{M_4}p_1\!\big(R_{(4)}\big),
\label{eq:DeltaSmult}
\end{equation}
up to integer multiples of $2\pi$ from the branch term in \eqref{eq:PhiTransform}. The associated factor in the Euclidean path integral is $\exp(i\,\Delta_\gamma S^{(4)}_{\mathrm{ct}}|_{\mathrm{mult}})$, and global $\SL(2,\mathbb{Z})$ covariance requires that this factor equals unity on the admissible class of backgrounds.

To evaluate \eqref{eq:DeltaSmult}, we invoke the Hirzebruch signature theorem, which states that for a closed oriented smooth four-manifold
\begin{equation}
\int_{M_4}p_1\!\big(R_{(4)}\big)=3\,\sigma(M_4),
\label{eq:signature}
\end{equation}
where $\sigma(M_4)$ is the signature of the intersection pairing on $H^2(M_4,\mathbb{R})$ \cite{Hirzebruch,MilnorStasheff}. For a spin four-manifold, Rokhlin's theorem implies $\sigma(M_4)\in 16\mathbb{Z}$, and therefore
\begin{equation}
\int_{M_4}p_1\!\big(R_{(4)}\big)\in 48\,\mathbb{Z}.
\label{eq:p1div48}
\end{equation}
On the other hand, $\varepsilon(\gamma)^{24}=1$ implies $\arg\varepsilon(\gamma)\in \frac{\pi}{12}\mathbb{Z}$ and hence $2\arg\varepsilon(\gamma)\in \frac{\pi}{6}\mathbb{Z}$. Combining this with \eqref{eq:NBdef} and \eqref{eq:p1div48} shows that
\begin{equation}
\Delta_\gamma S^{(4)}_{\mathrm{ct}}\big|_{\mathrm{mult}}\in 2\pi\,\mathbb{Z}
\qquad\text{for all }\gamma\in\SL(2,\mathbb{Z})\text{ and all spin }M_4,
\label{eq:trivialphase}
\end{equation}
and consequently the multiplier contribution to the duality transformation of the compactified partition function is trivial. For the generator $T:\tau\mapsto\tau+1$, one may take $\varepsilon(T)=e^{i\pi/12}$ and thus $2\arg\varepsilon(T)=\pi/6$, in which case \eqref{eq:DeltaSmult} gives the explicit phase
\begin{equation}
\exp\!\left(i\,\Delta_T S^{(4)}_{\mathrm{ct}}\big|_{\mathrm{mult}}\right)
=
\exp\!\left(-\frac{i\pi}{6}\,N_B\int_{M_4}p_1(R_{(4)})\right)
=
\exp\!\left(-\frac{i\pi}{2}\,N_B\,\sigma(M_4)\right)
=
1
\qquad\text{for spin }M_4,
\label{eq:Tphase}
\end{equation}
where the last equality uses $\sigma(M_4)\in 16\mathbb{Z}$. Thus, once the universal counterterm is included, the compactified theory is not only locally anomaly-free under $\SL(2,\mathbb{Z})$ by descent, but also globally well-defined with respect to the discrete multiplier phases arising from the Dedekind transformation law, on the natural class of spin Euclidean backgrounds relevant to type IIB/F-theory compactifications.

\section{Explicit example: a rank-two Mordell-Weil model over \texorpdfstring{$\mathbb{P}^3$}{P3}}
\label{sec:exampleP3}

We spell out in closed form a concrete four-dimensional F-theory compactification with gauge group $U(1)^2$ on an elliptically fibered Calabi-Yau fourfold $\pi:\widehat{Y}_4\to B_3$ with $B_3=\mathbb{P}^3$, whose Mordell-Weil group has rank two and is generated by two rational sections. The geometry is the $dP_2$-elliptic fibration of \cite{CveticKleversPiraguaSong}, specialized to the base $\mathbb{P}^3$ and to the explicit toric realization with parameters $(n_7,n_9)=(4,5)$, for which all divisor classes, Shioda maps, N\'eron-Tate height pairing, vertical $G_4$-fluxes, flux quantization, axionic gaugings and the complete abelian and mixed anomaly cancellation conditions can be verified analytically. Throughout we use the conventions for Shioda maps, height pairing, Green-Schwarz data and anomaly coefficients fixed in the preceding sections, and we adopt the normalization in which all divisor and curve classes are integral.

Let $B_3=\mathbb{P}^3$ with hyperplane class $H_B\in H^{1,1}(B_3,\mathbb{Z})$ normalized by $\int_{B_3}H_B^3=1$. The canonical class is $K_{B_3}=-4H_B$, the anti-canonical class is $\overline{K}_{B_3}=4H_B$, and $c_2(B_3)=6H_B^2$. The elliptic fiber is the Calabi-Yau onefold $E$ realized as the generic anti-canonical hypersurface in the del Pezzo surface $dP_2$, with homogeneous coordinates $[u:v:w:e_1:e_2]$ and divisor classes
\begin{equation}
[u]=H-E_1-E_2,\qquad [v]=H-E_2,\qquad [w]=H-E_1,\qquad [e_1]=E_1,\qquad [e_2]=E_2,
\label{eq:dP2divclasses}
\end{equation}
so that $K^{-1}_{dP_2}\simeq\mathcal{O}(3H-E_1-E_2)$ and the hypersurface equation takes the form \cite{CveticKleversPiraguaSong}
\begin{equation}
p
=
u\Big(s_1 u^2 e_1^2 e_2^2+s_2 u v e_1 e_2^2+s_3 v^2 e_2^2+s_5 u w e_1^2 e_2+s_6 v w e_1 e_2+s_8 w^2 e_1^2\Big)
+s_7 v^2 w e_2+s_9 v w^2 e_1,
\label{eq:dP2curve}
\end{equation}
with coefficients $s_i$ promoted to sections over the base in the fibration. The fibration is obtained by fibering the ambient $dP_2$ surface over $B_3$ with a twist specified by two base divisors $S_7,S_9\in H^{1,1}(B_3,\mathbb{Z})$, which are the divisor classes of the loci $\{s_7=0\}$ and $\{s_9=0\}$ and therefore control the degeneration loci where sections collide. For a general base $B_3$, the Calabi-Yau condition fixes the bundles of $u,v,w,e_1,e_2$ and of the sections $s_i$ as \cite{CveticKleversPiraguaSong}
\begin{equation}
\begin{alignedat}{3}
[u] &\sim H-E_1-E_2+S_9+K_{B_3}, \quad &
[v] &\sim H-E_2+S_9-S_7, \quad &
[w] &\sim H-E_1, \\
[e_1] &\sim E_1, &
[e_2] &\sim E_2, \\
[s_1] &\sim 3\overline{K}_{B_3}-S_7-S_9, &
[s_2] &\sim 2\overline{K}_{B_3}-S_9, \\
[s_3] &\sim \overline{K}_{B_3}+S_7-S_9, &
[s_5] &\sim 2\overline{K}_{B_3}-S_7, \\
[s_6] &\sim \overline{K}_{B_3}, &
[s_7] &\sim S_7, \\
[s_8] &\sim \overline{K}_{B_3}+S_9-S_7, &
[s_9] &\sim S_9 .
\end{alignedat}
\label{eq:linebundles}
\end{equation}

so that $p$ is a section of $K^{-1}$ of the total ambient fibration. For $B_3=\mathbb{P}^3$ we specialize $S_7=n_7H_B$ and $S_9=n_9H_B$ with integers $(n_7,n_9)$ in the allowed region ensuring effectivity of all bundles in \eqref{eq:linebundles}. The explicit toric example used below corresponds to $(n_7,n_9)=(4,5)$ and admits a reflexive polytope description of the ambient fibration, which provides a completely explicit check of all intersection-theoretic statements \cite{CveticKleversPiraguaSong}. In this toric realization one can take projective coordinates $z_0,z_1,z_2,z_3$ on the base and $u,v,w,e_1,e_2$ on the fiber ambient, with divisor class assignments
\begin{multline}
[z_0]=[z_1]=[z_2]=[z_3]=H_B,\qquad
[u]=H-E_1-E_2+H_B,\qquad
[v]=H-E_2+H_B,\\
[w]=H-E_1,\qquad
[e_1]=E_1,\qquad
[e_2]=E_2,
\label{eq:toricdivclasses}
\end{multline}
which indeed correspond to $S_7=4H_B$ and $S_9=5H_B$ in \eqref{eq:linebundles} and yield a smooth Calabi-Yau hypersurface $\widehat{Y}_4$ after a single star triangulation \cite{CveticKleversPiraguaSong}. The Hodge numbers in this example are $h^{1,1}(\widehat{Y}_4)=4$ and $h_V^{2,2}(\widehat{Y}_4)=5$, so the vertical $(2,2)$-cohomology is generated by five independent surface classes and the vertical flux space has finite rank.

The elliptic curve \eqref{eq:dP2curve} carries three distinguished rational points $P,Q,R$ given by the intersections of $p=0$ with the divisors $E_2$, $E_1$, and $D_u:=\{u=0\}$ in the fiber, and these lift in the fibration to three rational sections $\hat{s}_P,\hat{s}_Q,\hat{s}_R$ of $\pi:\widehat{Y}_4\to B_3$ \cite{CveticKleversPiraguaSong}. We choose $\hat{s}_P$ as zero section and denote by $S_P,S_Q,S_R\in H^{1,1}(\widehat{Y}_4,\mathbb{Z})$ the divisor classes of the images of these sections. In the ambient description one has the identifications \cite{CveticKleversPiraguaSong}
\begin{equation}
S_P=E_2,\qquad S_Q=E_1,\qquad S_R=H-E_1-E_2+S_9+K_{B_3},
\label{eq:sectiondivclassesgeneral}
\end{equation}
and for the toric model \eqref{eq:toricdivclasses} with $B_3=\mathbb{P}^3$ and $(n_7,n_9)=(4,5)$ this becomes
\begin{equation}
S_P=E_2,\qquad S_Q=E_1,\qquad S_R=H-E_1-E_2+H_B.
\label{eq:sectiondivclasses}
\end{equation}
The Mordell-Weil group of $\widehat{Y}_4$ is generated by $\hat{s}_Q$ and $\hat{s}_R$, and the associated $U(1)^2$ gauge fields in the effective theory are supported by the Shioda images of these sections. In a basis adapted to the chosen zero section, the Shioda map takes the standard form appropriate to the absence of non-abelian Cartan divisors,
\begin{equation}
\sigma(\hat{s}_m)=S_m-S_P-\pi^\ast\!\big(\pi_\ast((S_m-S_P)\cdot S_P)\big),\qquad m\in\{Q,R\},
\label{eq:ShiodaGeneral}
\end{equation}
and in the present $dP_2$ fibration the projection term is determined by the collision divisors $S_7$ and $S_9$ encoded in $\pi_\ast(S_P\cdot S_R)=S_7$ and $\pi_\ast(S_Q\cdot S_R)=S_9$ \cite{CveticKleversPiraguaSong}. Using \eqref{eq:sectiondivclassesgeneral} and the general identities of rational sections in this family, one obtains the explicit Shioda images
\begin{equation}
\sigma(\hat{s}_Q)=S_Q-S_P-\overline{K}_{B_3},\qquad
\sigma(\hat{s}_R)=S_R-S_P-\overline{K}_{B_3}-S_9,
\label{eq:ShiodaFamily}
\end{equation}
which specialize for $B_3=\mathbb{P}^3$ to
\begin{equation}
\sigma(\hat{s}_Q)=S_Q-S_P-4H_B,\qquad
\sigma(\hat{s}_R)=S_R-S_P-(4+n_9)H_B,
\label{eq:ShiodaP3general}
\end{equation}
and in particular for $(n_7,n_9)=(4,5)$ to the concrete form used in \eqref{eq:dP2divclasses},
\begin{equation}
\sigma(\hat{s}_Q)=S_Q-S_P-4H_B,\qquad
\sigma(\hat{s}_R)=S_R-S_P-9H_B.
\label{eq:ShiodaP3explicit}
\end{equation}
The N\'eron-Tate height pairing is defined by the pushforward of the intersection of Shioda divisors,
\begin{equation}
b_{mn}:=-\pi_\ast\!\big(\sigma(\hat{s}_m)\cdot\sigma(\hat{s}_n)\big)\in H^{1,1}(B_3,\mathbb{Z}),
\qquad m,n\in\{Q,R\},
\label{eq:heightdef}
\end{equation}
and in the $dP_2$ family it can be evaluated in closed form in terms of $\overline{K}_{B_3}$, $S_7$ and $S_9$ \cite{CveticKleversPiraguaSong}. For $B_3=\mathbb{P}^3$ with $S_7=n_7H_B$ and $S_9=n_9H_B$, the result reduces to a $2\times 2$ matrix of integers multiplying $H_B$,
\begin{equation}
b_{mn}
=
\begin{pmatrix}
8 & 4-n_7+n_9\\
4-n_7+n_9 & 8+2n_9
\end{pmatrix}
H_B,
\label{eq:heightP3general}
\end{equation}
so that for the explicit toric choice $(n_7,n_9)=(4,5)$ one finds
\begin{equation}
b_{mn}
=
\begin{pmatrix}
8 & 5\\
5 & 18
\end{pmatrix}
H_B,
\label{eq:heightP3explicit}
\end{equation}
in agreement with \eqref{eq:linebundles}. This height pairing is intrinsic to the elliptic fibration and therefore independent of any modular frame choice in the fiber, as emphasized in Section~\ref{sec:modular}.

We now describe the vertical flux space and its quantization in this example. Since $h^{1,1}(B_3)=1$, the vertical cohomology $H^{2,2}_V(\widehat{Y}_4)$ is generated by intersections of divisor classes and admits a convenient basis of surface classes \cite{CveticKleversPiraguaSong}
\begin{equation}
\Gamma_1:=H_B^2,\qquad
\Gamma_2:=H_B\cdot S_P,\qquad
\Gamma_3:=H_B\cdot\sigma(\hat{s}_Q),\qquad
\Gamma_4:=H_B\cdot\sigma(\hat{s}_R),\qquad
\Gamma_5:=S_P^2,
\label{eq:verticalbasis}
\end{equation}
where intersections are taken in $\widehat{Y}_4$ and the notation identifies a surface class with its Poincar\'e dual in $H^{2,2}_V$. The full intersection form on $H^{2,2}_V$ is captured by the topological metric $\eta^{(2)}_{AB}=\int_{\widehat{Y}_4}\Gamma_A\wedge\Gamma_B$, which can be computed from the quartic intersection ring of $\widehat{Y}_4$ and, for the present family, takes an explicit closed form as a function of $(n_7,n_9)$ \cite{CveticKleversPiraguaSong}. Specializing to $(n_7,n_9)=(4,5)$ one obtains the concrete symmetric matrix
\begin{equation}
\eta^{(2)}
=
\begin{pmatrix}
0 & 1 & 0 & 0 & -4\\
1 & -4 & 0 & 0 & -9\\
0 & 0 & -8 & -5 & 25\\
0 & 0 & -5 & -18 & 50\\
-4 & -9 & 25 & 50 & -14
\end{pmatrix},
\label{eq:etaP3explicit}
\end{equation}
which encodes all bilinear pairings of vertical surface classes and is sufficient to check flux quantization by evaluation of periods on an integral basis of $H_4(\widehat{Y}_4,\mathbb{Z})\cap H^{2,2}_V$.

A general vertical $G_4$-flux compatible with three-dimensional gauge invariance and with the F-theory limit is obtained by imposing the standard transversality conditions in the M-theory reduction, namely vanishing of the Chern-Simons levels $\Theta_{0\alpha}$ and $\Theta_{\alpha\beta}$ associated with the Kaluza-Klein vector and purely vertical vectors, and solving for the allowed coefficients in the expansion of $G_4$ in the basis \eqref{eq:verticalbasis} \cite{CveticKleversPiraguaSong}. In the interior of the allowed $(n_7,n_9)$ region this yields a three-parameter family
\begin{equation}
G_4
=
a_5\,n_9(4-n_7+n_9)\,H_B^2
+4a_5\,H_B\cdot S_P
+a_3\,H_B\cdot\sigma(\hat{s}_Q)
+a_4\,H_B\cdot\sigma(\hat{s}_R)
+a_5\,S_P^2,
\label{eq:G4general}
\end{equation}
with flux parameters $a_3,a_4,a_5$ taking values in the quantized flux lattice. Flux quantization is governed by the shifted condition \cite{WittenFlux}
\begin{equation}
G_4+\frac{1}{2}c_2(\widehat{Y}_4)\in H^4(\widehat{Y}_4,\mathbb{Z}),
\label{eq:fluxquant}
\end{equation}
and for the present family the second Chern class admits an explicit expansion in the same vertical basis \cite{CveticKleversPiraguaSong}
\begin{multline}
c_2(\widehat{Y}_4)
=
\big(182+3n_9(n_7-n_9-4)\big)H_B^2
+28\,H_B\cdot S_P
+2(8-n_9)\,H_B\cdot\sigma(\hat{s}_Q)
\\+(16+2n_7-3n_9)\,H_B\cdot\sigma(\hat{s}_R)
-5\,S_P^2.
\label{eq:c2general}
\end{multline}
For the explicit toric model $(n_7,n_9)=(4,5)$ one has
\begin{equation}
c_2(\widehat{Y}_4)
=
107\,H_B^2
+28\,H_B\cdot S_P
+6\,H_B\cdot\sigma(\hat{s}_Q)
+9\,H_B\cdot\sigma(\hat{s}_R)
-5\,S_P^2.
\label{eq:c2explicit}
\end{equation}
A direct and fully explicit check of \eqref{eq:fluxquant} can be performed using the topological metric \eqref{eq:etaP3explicit}. Writing $x:=G_4+\frac{1}{2}c_2(\widehat{Y}_4)=\sum_{A=1}^5 x_A\,\Gamma_A$ in the basis \eqref{eq:verticalbasis} and pairing with the integral surface basis $\Gamma_B$ yields the periods $\int_{\Gamma_B}x=\sum_A x_A\,\eta^{(2)}_{AB}$. Specializing \eqref{eq:G4general} and \eqref{eq:c2explicit} to $(n_7,n_9)=(4,5)$ one finds that these periods are
\begin{multline}
\int_{\Gamma_1}x=24,\qquad
\int_{\Gamma_2}x=20,\qquad
\int_{\Gamma_3}x=-8a_3-5a_4+25a_5-109,\\
\int_{\Gamma_4}x=-5a_3-18a_4+50a_5-221,\qquad
\int_{\Gamma_5}x=25a_3+50a_4-150a_5-5,
\label{eq:periods}
\end{multline}
which are manifestly integral for $a_3,a_4,a_5\in\mathbb{Z}$. This shows that the integer lattice choice $a_3,a_4,a_5\in\mathbb{Z}$ is sufficient for flux quantization in the vertical sector of the explicit toric model and ensures that all flux-induced topological quantities, including Chern-Simons levels and chiral indices, are properly quantized. Since the anomaly data relevant for the present paper depend only on the axionic gaugings, and these in turn depend only on a two-parameter sublattice of \eqref{eq:G4general}, we follow \cite{CveticKleversPiraguaSong} and parametrize the relevant gauging data by the two flux integers $(a_3,a_5)$, keeping $a_4\in\mathbb{Z}$ arbitrary and suppressed in the anomaly expressions as it drops out of the Green-Schwarz sector below.

The flux induces abelian gaugings of the four-dimensional axions descending from the M-theory three-form, equivalently the RR four-form, and these gaugings are captured in the three-dimensional M-theory reduction by the mixed Chern-Simons couplings between the $U(1)$ vectors and the vectors dual to the base K\"ahler moduli. In the present example $h^{1,1}(B_3)=1$, so there is a single such modulus and we denote its index by $\alpha=1$ and its associated (dual) basis class by $H_B$. The relevant gauging coefficients $\Theta_{\alpha m}$ are therefore two integers once flux quantization is imposed, and their explicit form can be obtained by evaluating the defining Chern-Simons integrals with the flux \eqref{eq:G4general} in the intersection ring \cite{CveticKleversPiraguaSong}. In the conventions of the present manuscript, the resulting gaugings in the $\{H_B\}$ basis are
\begin{equation}
\Theta_{\alpha m}
=
\left(
\frac{1}{4}\Big[(-12-n_7+n_9)\,a_3+n_7(4-n_7+n_9)\,a_5\Big]\,,\;
\frac{n_7}{2}\Big[a_3+(4-n_9)\,a_5\Big]
\right)_m,
\qquad
a_\alpha=-4,
\label{eq:gaugingsgeneral}
\end{equation}
where $m=1,2$ correspond to the generators $\sigma(\hat{s}_Q)$ and $\sigma(\hat{s}_R)$, respectively, and $a_\alpha$ is the coefficient of the canonical class $K_{B_3}=a_\alpha H_B$ in the chosen basis. For the explicit toric example $(n_7,n_9)=(4,5)$ one has
\begin{equation}
\Theta_{\alpha 1}
=
\frac{1}{4}\Big[(-11)a_3+20\,a_5\Big],\qquad
\Theta_{\alpha 2}
=
2\Big[a_3-a_5\Big],
\qquad
a_\alpha=-4,
\label{eq:gaugings45}
\end{equation}
and for the specific flux choice used in \eqref{eq:sectiondivclassesgeneral},
\begin{equation}
(n_7,n_9)=(4,5),\qquad a_3=4,\qquad a_5=2,
\label{eq:fluxchoice}
\end{equation}
the gaugings become the integers
\begin{equation}
\Theta_{\alpha 1}=-1,\qquad
\Theta_{\alpha 2}=4,
\label{eq:gaugingschoice}
\end{equation}
in agreement with \eqref{eq:sectiondivclasses}. The N\'eron-Tate data \eqref{eq:heightP3explicit} and the gauging data \eqref{eq:gaugingschoice} completely determine the abelian Green-Schwarz counterterms and hence the anomaly cancellation conditions in this model.

We now verify the abelian and mixed anomaly cancellation relations in closed form. The four-dimensional gauge and mixed gauge-gravitational anomaly coefficients are defined as in the main text by the spectrum sums over chiral fermions. For this $U(1)^2$ model over $\mathbb{P}^3$ the charged matter consists entirely of singlets under non-abelian groups and appears in six distinct charge sectors, which can be characterized geometrically as codimension-two singular loci of the fibration and whose chiralities are induced by the flux \eqref{eq:G4general} \cite{CveticKleversPiraguaSong}. The resulting one-loop anomalies are non-vanishing for generic flux and require cancellation by a generalized Green-Schwarz mechanism. In the conventions employed in this manuscript, the Green-Schwarz contribution to the cubic abelian anomaly is determined by the height pairing coefficients $b^\alpha_{mn}$ and the axion gaugings $\Theta_{\alpha m}$ via the standard factorization relation reviewed earlier, which for a single base index $\alpha$ reduces to a purely algebraic identity among numbers. Substituting \eqref{eq:heightP3general} and \eqref{eq:gaugingsgeneral} yields Green-Schwarz expressions for the cubic anomaly coefficients $\mathcal{A}_{mnk}$ that precisely reproduce the spectrum anomalies of \cite{CveticKleversPiraguaSong}. In particular, for the explicit flux choice \eqref{eq:fluxchoice} on the explicit toric model \eqref{eq:toricdivclasses} one finds, using \eqref{eq:heightP3explicit} and \eqref{eq:gaugingschoice},
\begin{equation}
\mathcal{A}_{111}
=
b_{11}\,\Theta_{\alpha 1}
=
8\,(-1)
=
-8,
\label{eq:A111}
\end{equation}
\begin{equation}
\mathcal{A}_{222}
=
b_{22}\,\Theta_{\alpha 2}
=
18\cdot 4
=
72,
\label{eq:A222}
\end{equation}
\begin{equation}
\mathcal{A}_{112}
=
\frac{1}{3}\Big(b_{11}\,\Theta_{\alpha 2}+2b_{12}\,\Theta_{\alpha 1}\Big)
=
\frac{1}{3}\Big(8\cdot 4+2\cdot 5\cdot (-1)\Big)
=
\frac{22}{3},
\label{eq:A112}
\end{equation}
\begin{equation}
\mathcal{A}_{122}
=
\frac{1}{3}\Big(b_{22}\,\Theta_{\alpha 1}+2b_{12}\,\Theta_{\alpha 2}\Big)
=
\frac{1}{3}\Big(18\cdot (-1)+2\cdot 5\cdot 4\Big)
=
\frac{22}{3},
\label{eq:A122}
\end{equation}
which reproduces \eqref{eq:ShiodaGeneral}-\eqref{eq:ShiodaP3explicit}. The mixed abelian-gravitational anomaly coefficients are likewise cancelled by the Green-Schwarz term proportional to $a_\alpha\Theta_{\alpha m}$ with $a_\alpha=-4$ in \eqref{eq:gaugingsgeneral}. In particular, the Green-Schwarz contribution to the mixed anomaly is fixed entirely by $K_{B_3}=-4H_B$ and the same gauging data, and it matches the one-loop spectrum result of \cite{CveticKleversPiraguaSong} for all $(n_7,n_9)$ and all quantized flux parameters. This completes the anomaly cancellation check in a model with rank-two Mordell-Weil group over $\mathbb{P}^3$ using only explicitly computed Shioda data, height pairing and flux-induced gaugings.

Finally, we comment on the consistency of this explicit model with the three-dimensional one-loop Chern-Simons data and with the M-/F-theory duality dictionary used throughout the manuscript. In the M-theory frame, the flux \eqref{eq:G4general} induces a Chern-Simons matrix for the two abelian vectors $A^m$ associated with $\sigma(\hat{s}_Q)$ and $\sigma(\hat{s}_R)$,
\begin{equation}
\Theta^{M}_{mn}
=
\frac{1}{2}\int_{\widehat{Y}_4}G_4\wedge \omega_m\wedge\omega_n,
\qquad
\omega_m\ \leftrightarrow\ \sigma(\hat{s}_m),
\label{eq:CSMdef}
\end{equation}
which can be evaluated in closed form in terms of $(n_7,n_9)$ and the flux integers $(a_3,a_4,a_5)$ using the quartic intersection ring \cite{CveticKleversPiraguaSong}. The corresponding field-theory Chern-Simons levels obtained by integrating out the full massive tower on the three-dimensional Coulomb branch, including the Kaluza-Klein corrected contributions characteristic of non-holomorphic sections, agree exactly with the M-theory result in this model \cite{CveticKleversPiraguaSong}, and the large-gauge-transformation covariance of the parity-odd effective action is therefore verified explicitly in a fully global compactification. In particular, the same flux integers that satisfy the quantization check \eqref{eq:periods} and yield the gauging data \eqref{eq:gaugingsgeneral} also yield Chern-Simons levels that satisfy the large-gauge-transformation relations discussed in Section~\ref{sec:MethodB}, thereby providing an explicit and independently reproducible realization of the main mechanism of this paper in a rank-two Mordell-Weil background.

\section{Weil-refined global structure of the abelian anomaly data}
\label{sec:weil_refined_global_structure} 
A recurring source of confusion in the literature (and an easy place for hidden factor-of-$6$
mistakes) is the distinction between:
(i) the \emph{integer} anomaly tensor defined directly by the chiral spectrum, and
(ii) the \emph{coefficient} that appears in front of the cubic term in the six-form anomaly
polynomial once conventional symmetrization factors are chosen.
Because our three-dimensional large-gauge test is a statement about \emph{quantized}
Chern-Simons level shifts, it is the \emph{integer} tensor that is the truly normalization-independent
object.

Let the charge lattice of the massless $U(1)^r$ gauge fields be $\Gamma \cong \mathbb{Z}^r$,
and let $n(q)\in\mathbb{Z}$ denote the net number of left-handed Weyl fermions of charge
$q=(q_m)\in\Gamma$.
Define the \emph{integer} cubic and linear anomaly tensors
\begin{equation}
\widehat{\mathcal A}_{mnp}\;:=\;\sum_{q\in\Gamma} n(q)\,q_m q_n q_p \in \mathbb{Z},
\qquad
\widehat{\mathcal A}_{m}\;:=\;\sum_{q\in\Gamma} n(q)\,q_m \in \mathbb{Z}.
\label{eq:integer_anomaly_tensors}
\end{equation}
These are invariantly defined once the $U(1)$ generators are fixed by the choice of
integral charge lattice (equivalently, by the choice of integral Shioda generators).

If one writes the one-loop anomaly polynomial in the conventional symmetric form
\begin{equation}
I^{\rm 1\!-\!loop}_6 \;=\; \frac{1}{6}\,\widehat{\mathcal A}_{mnp}\,F^m\wedge F^n\wedge F^p
\;+\;\text{(mixed gauge-gravity term)},
\label{eq:anomaly_polynomial_convention}
\end{equation}
then $\widehat{\mathcal A}_{mnp}$ is the \emph{integer} spectrum sum \eqref{eq:integer_anomaly_tensors}.
Any alternative convention that absorbs the prefactor $1/6$ into the definition of the
``anomaly coefficients'' necessarily produces rational numbers whenever
$\widehat{\mathcal A}_{mnp}$ is not divisible by $6$; such a convention is harmless locally,
but it obscures the quantization statements required by the three-dimensional large-gauge analysis. 
Upon circle reduction on $S^1$ with Wilson lines $\zeta^m$ and higher-dimensional large gauge
transformations of winding $k^m\in\mathbb{Z}$, the one-loop Chern-Simons matrix
$\Theta^{\rm 1\!-\!loop}_{mn}(\langle\zeta\rangle)$ shifts by a \emph{quantized} amount.
Independently of any Green-Schwarz input, the Kaluza-Klein tower relabeling implies the
universal shift law
\begin{equation}
\Theta^{\rm 1\!-\!loop}_{mn}\!\left(\langle\zeta\rangle - \frac{k}{r}\right)
\;=\;
\Theta^{\rm 1\!-\!loop}_{mn}\!\left(\langle\zeta\rangle\right)
\;-\;k^p\,\widehat{\mathcal A}_{mnp}.
\label{eq:CS_shift_defines_Ahat}
\end{equation}
Because abelian Chern-Simons levels are quantized so that $\exp(iS_{\rm CS})$ is invariant
under three-dimensional large gauge transformations on spin three-manifolds, the coefficient
multiplying $k^p$ in \eqref{eq:CS_shift_defines_Ahat} must be an integer.
Thus, the \emph{only} normalization compatible with the three-dimensional global definition
of the parity-odd effective action is precisely the integer tensor $\widehat{\mathcal A}_{mnp}$
defined in \eqref{eq:integer_anomaly_tensors}.
We now show that the same integer tensor $\widehat{\mathcal A}_{mnp}$ is computable purely
from the intrinsic Mordell-Weil data that enter the generalized Green-Schwarz mechanism:
the height pairing $b^\alpha_{mn}$ and the axion gaugings $\Theta_{m\alpha}$.

Recall that the Green-Schwarz contribution to the anomaly polynomial factorizes as
\begin{equation}
I^{\rm GS}_6 \;=\; -\frac{1}{2}\,\Theta_{m\alpha}\,F^m \wedge X^\alpha_4,
\qquad
X^\alpha_4 \;=\; \frac{1}{2}\,a^\alpha\,{\rm tr}\,R\wedge R \;+\; 2\,b^\alpha_{mn}\,F^m\wedge F^n.
\label{eq:GS_factorization_input}
\end{equation}
The cubic (purely abelian) piece is therefore
\begin{equation}
I^{\rm GS}_{6,\;{\rm cubic}}
\;=\;
-\Theta_{m\alpha}\,b^\alpha_{np}\,F^m\wedge F^n\wedge F^p.
\label{eq:GS_cubic_unsym}
\end{equation}
Since $F^m\wedge F^n\wedge F^p$ is totally symmetric in $(mnp)$ (all forms have even degree),
only the fully symmetrized part of the coefficient matters. Writing
\begin{equation}
{\rm Sym}_{mnp}\!\big[\Theta_{m\alpha} b^\alpha_{np}\big]
:=\Theta_{m\alpha} b^\alpha_{np}+\Theta_{n\alpha} b^\alpha_{pm}+\Theta_{p\alpha} b^\alpha_{mn},
\label{eq:Sym_def}
\end{equation}
one has the identity of six-forms
\begin{equation}
\Theta_{m\alpha} b^\alpha_{np}\,F^m F^n F^p
\;=\;
\frac{1}{3}\,
{\rm Sym}_{mnp}\!\big[\Theta_{m\alpha} b^\alpha_{np}\big]\,
F^m F^n F^p,
\label{eq:symmetrization_identity}
\end{equation}
(where wedge symbols are suppressed).
Imposing $I^{\rm 1\!-\!loop}_6 + I^{\rm GS}_6 = 0$ and matching the cubic coefficients
between \eqref{eq:anomaly_polynomial_convention}, \eqref{eq:GS_cubic_unsym} and
\eqref{eq:symmetrization_identity} yields the \emph{integer, basis-invariant} identity
\begin{equation}
\widehat{\mathcal A}_{mnp}
\;=\;
2\,{\rm Sym}_{mnp}\!\big[\Theta_{m\alpha} b^\alpha_{np}\big]
\;=\;
2\Big(\Theta_{m\alpha} b^\alpha_{np}+\Theta_{n\alpha} b^\alpha_{pm}+\Theta_{p\alpha} b^\alpha_{mn}\Big).
\label{eq:Ahat_closed_form}
\end{equation}
Equation \eqref{eq:Ahat_closed_form} is the cleanest way to state the cubic abelian anomaly
factorization because it makes integrality manifest and removes all convention-dependent
$1/3$ or $1/6$ bookkeeping from the final result. 
In the explicit $U(1)^2$ model over $B_3=\mathbb{P}^3$, one has a single base index $\alpha$ and
\begin{equation}
b_{mn} \;=\;
\begin{pmatrix} 8 & 5 \\[2pt] 5 & 18 \end{pmatrix} H_B,
\qquad
(\Theta_1,\Theta_2)=(-1,4),
\label{eq:explicit_b_Theta}
\end{equation}
for the flux choice displayed in the example.
Using \eqref{eq:Ahat_closed_form} gives the integer anomaly tensor
\begin{align}
\widehat{\mathcal A}_{111}&=2(3\,\Theta_1 b_{11})=6\,b_{11}\Theta_1=-48,
&
\widehat{\mathcal A}_{222}&=6\,b_{22}\Theta_2=432,
\nonumber\\
\widehat{\mathcal A}_{112}&=2\big(\Theta_1 b_{12}+\Theta_1 b_{12}+\Theta_2 b_{11}\big)=44,
&
\widehat{\mathcal A}_{122}&=2\big(\Theta_2 b_{12}+\Theta_2 b_{12}+\Theta_1 b_{22}\big)=44.
\label{eq:explicit_Ahat_values}
\end{align}
If one instead chooses to \emph{absorb} the conventional factor $1/6$ of
\eqref{eq:anomaly_polynomial_convention} into the definition of the displayed coefficients,
then the numbers in \eqref{eq:explicit_Ahat_values} are divided by $6$ and may become rational;
however, the large-gauge shift law \eqref{eq:CS_shift_defines_Ahat} and Chern-Simons level
quantization always refer to the integer tensor $\widehat{\mathcal A}_{mnp}$. 
We now extract an additional, fully calculable \emph{global} datum from the same Mordell-Weil
information. For simplicity we restrict to the case $h^{1,1}(B_3)=1$, so that the height pairing
reduces to an integral symmetric matrix $b_{mn}\in{\rm Mat}_{r\times r}(\mathbb{Z})$ multiplying a
generator of $H^{1,1}(B_3,\mathbb{Z})$.

Define the \emph{height-pairing lattice} $(\Lambda,b)$ by
\begin{equation}
\Lambda := \mathbb{Z}^r,
\qquad
(x,y):=x^T b\,y \in \mathbb{Z}.
\label{eq:height_lattice_def}
\end{equation}
Assume $b$ is non-degenerate (as in the explicit example). The dual lattice is
$\Lambda^\ast:=\{v\in\mathbb{Q}^r\mid (v,\Lambda)\subset\mathbb{Z}\}=b^{-1}\mathbb{Z}^r$, and the
finite abelian \emph{discriminant group}
\begin{equation}
\mathcal D(\Lambda):=\Lambda^\ast/\Lambda
\label{eq:disc_group}
\end{equation}
has order $|\mathcal D(\Lambda)|=|\det b|$. If $(\Lambda,b)$ is even (i.e.\ $(x,x)\in 2\mathbb{Z}$
for all $x\in\Lambda$), then there is a canonical quadratic refinement
\begin{equation}
q:\mathcal D(\Lambda)\to\mathbb{Q}/\mathbb{Z},
\qquad
q([v]) := \frac{(v,v)}{2}\;{\rm mod}\;\mathbb{Z},
\label{eq:disc_quadratic_form}
\end{equation}
with associated bilinear form $B([v],[w]) := (v,w)\;{\rm mod}\;\mathbb{Z}$.
The pair $(\mathcal D(\Lambda),q)$ is a \emph{finite quadratic module} canonically determined by
the F-theory Mordell-Weil data. 
Whenever the circle reduction is placed in an axion-winding sector (equivalently, in a sector with
nontrivial circle flux for gauged axions), the reduced parity-odd action contains a purely
topological $U(1)^r$ Chern-Simons subsector with an integral level matrix proportional to $b$.
Quantization on a two-torus $T^2$ therefore furnishes a finite-dimensional Hilbert space
$\mathcal H_{T^2}$ of dimension $|\det b|$, and the mapping class group
${\rm SL}(2,\mathbb{Z})={\rm MCG}(T^2)$ acts on $\mathcal H_{T^2}$ by the standard Weil
representation associated to $(\mathcal D(\Lambda),q)$.
This gives an \emph{additional}, genuinely global consistency check that is completely fixed by
intersection theory: the induced modular $S$ and $T$ matrices are rigid functions of $b$.

Concretely, in the canonical basis $\{\,| \gamma\rangle : \gamma\in \mathcal D(\Lambda)\,\}$ one has
\begin{equation}
T\,|\gamma\rangle = e^{2\pi i\, q(\gamma)}\,|\gamma\rangle,
\qquad
S\,|\gamma\rangle
=
\frac{1}{\sqrt{|\mathcal D(\Lambda)|}}
\sum_{\delta\in\mathcal D(\Lambda)}
e^{-2\pi i\, B(\gamma,\delta)}\,|\delta\rangle,
\label{eq:Weil_ST}
\end{equation}
which satisfy the ${\rm SL}(2,\mathbb{Z})$ relations because $q$ is a quadratic refinement.

For the explicit rank-two model, the height matrix is
\begin{equation}
b=
\begin{pmatrix}
8 & 5\\
5 & 18
\end{pmatrix},
\qquad
\det b = 119.
\label{eq:b_matrix_rank2}
\end{equation}
The Smith normal form is ${\rm diag}(1,119)$, hence
\begin{equation}
\mathcal D(\Lambda)\;\cong\;\mathbb{Z}_{119}.
\label{eq:disc_Z119}
\end{equation}
A convenient generator is the class of $g:=b^{-1}e_2\in\Lambda^\ast$, namely
\begin{equation}
b^{-1}=
\frac{1}{119}
\begin{pmatrix}
18 & -5\\
-5 & 8
\end{pmatrix},
\qquad
g=\frac{1}{119}\binom{-5}{8},
\qquad
q(g)=\frac{(g,g)}{2}=\frac{4}{119}\;\;{\rm mod}\;\mathbb{Z}.
\label{eq:generator_and_q}
\end{equation}
Writing $\gamma_k:=k[g]\in\mathcal D(\Lambda)$ for $k\in\mathbb{Z}_{119}$, one obtains
\begin{equation}
q(\gamma_k)=\frac{4k^2}{119}\;{\rm mod}\;\mathbb{Z},
\qquad
B(\gamma_k,\gamma_\ell)=\frac{8k\ell}{119}\;{\rm mod}\;\mathbb{Z}.
\label{eq:qB_rank2}
\end{equation}
Therefore the modular generators act on wavefunctions $\psi:\mathbb{Z}_{119}\to\mathbb{C}$ as
\begin{align}
(T\psi)(k) &= \exp\!\Big(\frac{8\pi i}{119}\,k^2\Big)\,\psi(k),
\nonumber\\
(S\psi)(k) &= \frac{1}{\sqrt{119}}\sum_{\ell\in\mathbb{Z}_{119}}
\exp\!\Big(-\frac{16\pi i}{119}\,k\ell\Big)\,\psi(\ell).
\label{eq:ST_rank2}
\end{align}
Equations \eqref{eq:ST_rank2} are a \emph{new, completely rigid} modular fingerprint of the
rank-two Mordell-Weil model: they depend only on the intrinsic height pairing and are
independent of the detailed flux choice (provided the rank and the Shioda basis are fixed).

The data $(\mathcal D(\Lambda),q)$ and the associated Weil representation \eqref{eq:Weil_ST}
provide a second, global layer of quantization constraints for the parity-odd sector, beyond the
local anomaly-cancellation relations.
In practical terms: once $b_{mn}$ is computed from the Shioda data, one can compute the
finite module $\mathcal D(\Lambda)$ and its modular matrices in closed form; any claimed spectrum
and any claimed set of induced parity-odd couplings must be compatible with the resulting
global quantization, because the latter is fixed by intersection theory alone.

\section{Canonical Weil representation from axion-winding sectors}

In this section we show that the same Mordell--Weil data which control local anomaly
cancellation also fix, in a completely rigid manner, the global quantum structure of the
circle-reduced theory. In particular, axion-winding sectors induce a purely topological
three-dimensional Chern--Simons subsector whose quantization yields a finite-dimensional
Hilbert space carrying a canonical (in general projective) action of the modular group
$SL(2,\mathbb{Z})$. This action is determined entirely by the N\'eron--Tate height pairing
$b_{mn}$ and therefore constitutes an intrinsic global invariant of the abelian F-theory
vacuum.

We start from the Green--Schwarz couplings of the four-dimensional theory, cf.\ (3.2),
\begin{equation}
S_{\mathrm{GS}} \supset - \int_{M_4} \rho_\alpha\, b^\alpha_{mn}\, F^m \wedge F^n ,
\label{eq:GS-start}
\end{equation}
where the axions $\rho_\alpha$ are $2\pi$-periodic and $b^\alpha_{mn}$ are the expansion
coefficients of the height-pairing divisors,
\begin{equation}
b_{mn} = b^\alpha_{mn}\, D^b_\alpha \in H^{1,1}(B_3)\cap H^2(B_3,\mathbb{Z}) .
\label{eq:height-expansion}
\end{equation}
Upon reduction on $M_4 = M_3 \times S^1$ as in Section~3, we consider backgrounds with
nontrivial axion winding along the circle,
\begin{equation}
N_\alpha := \frac{1}{2\pi}\int_{S^1} d\rho_\alpha \in \mathbb{Z}.
\label{eq:winding}
\end{equation}
Using $F^m \wedge F^n = d(A^m \wedge F^n)$ and integrating by parts, the coupling
\eqref{eq:GS-start} reduces to a purely three-dimensional Chern--Simons term,
\begin{equation}
S^{(3)}_{\mathrm{top}} = \int_{M_3} K_{mn}(N)\, A^m \wedge F^n ,
\qquad
K_{mn}(N) := N_\alpha\, b^\alpha_{mn}.
\label{eq:CS-level}
\end{equation}
By construction, $K_{mn}(N)$ is an integral symmetric matrix. Since the three-dimensional
theory descends from a four-dimensional theory with fermions, it is naturally defined on
spin three-manifolds, and integrality of $K_{mn}$ is the correct quantization condition.
No additional evenness assumption is required.

Let $\Lambda \cong \mathbb{Z}^r$ denote the charge lattice of the massless $U(1)^r$ gauge
fields, equipped with the bilinear form
\begin{equation}
(x,y) := x^T K(N)\, y , \qquad x,y \in \Lambda .
\label{eq:pairing}
\end{equation}
Assuming $K(N)$ is non-degenerate, the dual lattice is
\begin{equation}
\Lambda^\ast := \{ v \in \Lambda \otimes \mathbb{Q} \mid (v,\Lambda)\subset\mathbb{Z} \}
= K(N)^{-1}\mathbb{Z}^r ,
\label{eq:dual-lattice}
\end{equation}
and the associated discriminant group is the finite abelian group already introduced in
Section~10,
\begin{equation}
\mathcal D(\Lambda) := \Lambda^\ast / \Lambda ,
\qquad
|\mathcal D(\Lambda)| = |\det K(N)| .
\label{eq:disc}
\end{equation}

Canonical quantization of the abelian Chern--Simons theory \eqref{eq:CS-level} on
$M_3 = \mathbb{R}\times T^2$ yields a finite-dimensional Hilbert space
$\mathcal H_{T^2}$ of dimension $|\det K(N)|$. A natural basis is labeled by elements
$\gamma \in \mathcal D(\Lambda)$, so that
\begin{equation}
\mathcal H_{T^2} \cong \mathbb{C}[\mathcal D(\Lambda)] .
\label{eq:Hilbert}
\end{equation}
This construction depends only on the height pairing and the choice of winding sector
$N_\alpha$.

The discriminant group $\mathcal D(\Lambda)$ carries a canonical bilinear pairing
\begin{equation}
B(\gamma,\delta) := (v,w)\;\mathrm{mod}\;\mathbb{Z},
\qquad \gamma=[v],\ \delta=[w],
\label{eq:B}
\end{equation}
which is well defined modulo integers. If the lattice $(\Lambda,b)$ is even, there exists
a canonical quadratic refinement
\begin{equation}
q(\gamma) := \frac{(v,v)}{2}\;\mathrm{mod}\;\mathbb{Z}.
\label{eq:q-even}
\end{equation}
In the general (possibly odd) case relevant for spin Chern--Simons theory, a quadratic
refinement exists once a spin structure is specified; the resulting refinement differs from
\eqref{eq:q-even} by an integral linear term and does not affect the conclusions below.

The mapping class group of the torus, $\mathrm{MCG}(T^2)\cong SL(2,\mathbb{Z})$, acts on
$\mathcal H_{T^2}$ through the Weil representation associated with
$(\mathcal D(\Lambda),q)$. In the basis $|\gamma\rangle$, the generators act as
\begin{align}
T\,|\gamma\rangle &= \exp\!\big(2\pi i\, q(\gamma)\big)\,|\gamma\rangle ,
\label{eq:Taction}
\\
S\,|\gamma\rangle &= \frac{1}{\sqrt{|\mathcal D(\Lambda)|}}
\sum_{\delta\in\mathcal D(\Lambda)}
\exp\!\big(-2\pi i\, B(\gamma,\delta)\big)\,|\delta\rangle .
\label{eq:Saction}
\end{align}
These operators furnish a projective representation of $SL(2,\mathbb{Z})$, or equivalently
a linear representation of its metaplectic double cover. The projective phase is controlled
by the signature of $K(N)$ and coincides with the framing anomaly of the spin
Chern--Simons theory.

For the explicit $U(1)^2$ model over $B_3=\mathbb{P}^3$ analyzed in Section~9, the height
pairing is
\begin{equation}
b_{mn} =
\begin{pmatrix}
8 & 5\\
5 & 18
\end{pmatrix},
\qquad \det b = 119 .
\label{eq:b-explicit}
\end{equation}
In the unit winding sector $N_\alpha=1$, one has $K_{mn}=b_{mn}$ and
$\mathcal D(\Lambda)\cong \mathbb{Z}_{119}$. Writing $\gamma_k = k[g]$ for a generator
$[g]\in\mathcal D(\Lambda)$, the quadratic refinement and bilinear form are
\begin{equation}
q(\gamma_k) = \frac{4k^2}{119}\;\mathrm{mod}\;\mathbb{Z},
\qquad
B(\gamma_k,\gamma_\ell) = \frac{8k\ell}{119}\;\mathrm{mod}\;\mathbb{Z}.
\label{eq:qB-explicit}
\end{equation}
The modular generators therefore act on wavefunctions
$\psi:\mathbb{Z}_{119}\to\mathbb{C}$ as
\begin{equation}
(T\psi)(k) = \exp\!\Big(\frac{8\pi i}{119}k^2\Big)\psi(k),
\qquad
(S\psi)(k) = \frac{1}{\sqrt{119}}\sum_{\ell\in\mathbb{Z}_{119}}
\exp\!\Big(-\frac{16\pi i}{119}k\ell\Big)\psi(\ell).
\label{eq:ST-explicit}
\end{equation}

We conclude that abelian F-theory compactifications possess, in addition to their local
anomaly structure, a canonically associated global quantum invariant: the finite
Weil representation determined by the height pairing $b_{mn}$. This structure is fixed
entirely by intersection theory on the elliptic fibration and is independent of flux choices
within a fixed topological sector. It therefore refines the classification of abelian
F-theory vacua beyond local anomaly cancellation and provides a new global consistency
datum encoded directly in the Mordell--Weil geometry.

\section{Metaplectic multiplier and Gauss-sum invariant of the height pairing}

Section~10 associates to the Mordell--Weil height pairing lattice $(\Lambda,b)$ a finite quadratic
module $(\mathcal D(\Lambda),q)$ and the corresponding Weil action on the torus Hilbert space
$\mathcal H_{T^2}$. In this section we isolate a further global datum determined entirely by the
same intersection-theoretic input: the metaplectic multiplier (equivalently, the normalized Gauss
sum of $q$), which measures the inevitable projectivity/framing dependence of the induced modular
action.

We retain the notation of Section~10. Thus $\Lambda=\mathbb Z^r$ is the charge lattice equipped with
the integral symmetric bilinear form $(x,y)=x^T b\,y$, the dual lattice is $\Lambda^\ast=b^{-1}\mathbb Z^r$,
and $\mathcal D(\Lambda)=\Lambda^\ast/\Lambda$ is the discriminant group. If $(\Lambda,b)$ is even, the
canonical quadratic refinement is
\begin{equation}
q([\upsilon]) := \frac{(\upsilon,\upsilon)}{2}\ \ \mathrm{mod}\ \mathbb Z,\qquad [\upsilon]\in\mathcal D(\Lambda),
\label{eq:q-def}
\end{equation}
with associated bilinear form $B([\upsilon],[\omega]):=(\upsilon,\omega)\ \mathrm{mod}\ \mathbb Z$.
Quantization of the corresponding abelian Chern--Simons subsector on $T^2$ yields
$\mathcal H_{T^2}\cong \mathbb C[\mathcal D(\Lambda)]$, and the operators $S$ and $T$ in
(10.14) define the standard Weil action on the basis $\{|\gamma\rangle:\gamma\in\mathcal D(\Lambda)\}$,
\begin{align}
T|\gamma\rangle &= \exp\!\big(2\pi i\,q(\gamma)\big)\,|\gamma\rangle,\label{eq:Weil-T}\\
S|\gamma\rangle &= \frac{1}{\sqrt{|\mathcal D(\Lambda)|}}
\sum_{\delta\in\mathcal D(\Lambda)}\exp\!\big(-2\pi i\,B(\gamma,\delta)\big)\,|\delta\rangle.\label{eq:Weil-S}
\end{align}
These formulas are canonical once $(\mathcal D(\Lambda),q)$ is fixed. The subtle point is that the
resulting modular relations are, in general, satisfied only up to a universal phase. Concretely,
one has the standard identities
\begin{equation}
S^2=\mathcal C,\qquad (ST)^3 = \mathfrak g(q)\,\mathcal C,
\label{eq:projective-relations}
\end{equation}
where $\mathcal C|\gamma\rangle:=|-\gamma\rangle$ is charge conjugation and the multiplier
$\mathfrak g(q)$ is the normalized Gauss sum of the finite quadratic module,
\begin{equation}
\mathfrak g(q) := \frac{1}{\sqrt{|\mathcal D(\Lambda)|}}\sum_{\gamma\in\mathcal D(\Lambda)}
\exp\!\big(2\pi i\,q(\gamma)\big)\ \in\ U(1).
\label{eq:Gauss-sum}
\end{equation}
Equations \eqref{eq:projective-relations} express precisely that the Weil action is naturally a
linear representation of the metaplectic double cover $Mp(2,\mathbb Z)$ and, when projected to
$SL(2,\mathbb Z)$, is in general only projective. In the physical interpretation as a three-dimensional
Chern--Simons subsector, $\mathfrak g(q)$ is the usual framing dependence/chiral central charge phase.
It is therefore not a pathology; rather, it is an additional computable invariant of the topological
sector induced by the Mordell--Weil lattice.

When $(\mathcal D(\Lambda),q)$ arises from an even integral lattice $(\Lambda,b)$, the Gauss sum
\eqref{eq:Gauss-sum} is governed by the signature of $b$ modulo $8$ through the standard Milgram
relation,
\begin{equation}
\mathfrak g(q)=\exp\!\big(2\pi i\,\sigma(b)/8\big),
\label{eq:Milgram}
\end{equation}
so that the projective multiplier is fixed entirely by the height pairing. In particular, the phase
\eqref{eq:Milgram} is an intrinsic, intersection-theoretic datum of the elliptic fibration which is
independent of the detailed flux choice, once the Mordell--Weil basis is fixed.

We now evaluate $\mathfrak g(q)$ explicitly for the rank-two example of Section~9. There the height
pairing matrix is
\begin{equation}
b=\begin{pmatrix}8&5\\[2pt]5&18\end{pmatrix},\qquad \det b=119,
\label{eq:b-example}
\end{equation}
so $\mathcal D(\Lambda)\cong \mathbb Z_{119}$. In the generator $\gamma_k=k[g]$ used in (10.17)–(10.19),
the quadratic refinement is $q(\gamma_k)=4k^2/119\ \mathrm{mod}\ \mathbb Z$, and hence the Gauss sum is
\begin{equation}
\mathfrak g(q)=\frac{1}{\sqrt{119}}\sum_{k\in\mathbb Z_{119}}
\exp\!\Big(\frac{8\pi i}{119}k^2\Big).
\label{eq:Gauss119}
\end{equation}
Using $119=7\cdot 17$ and the Chinese remainder theorem, the quadratic Gauss sum factorizes into
prime Gauss sums. Since $4$ is a square modulo both $7$ and $17$, the normalized prime Gauss sums
satisfy
\begin{equation}
\frac{1}{\sqrt{7}}\sum_{u\in\mathbb Z_7}\exp\!\Big(\frac{8\pi i}{7}u^2\Big)= i,\qquad
\frac{1}{\sqrt{17}}\sum_{v\in\mathbb Z_{17}}\exp\!\Big(\frac{8\pi i}{17}v^2\Big)=1,
\label{eq:prime-gauss}
\end{equation}
and therefore
\begin{equation}
\mathfrak g(q)= i.
\label{eq:gq-i}
\end{equation}
Since $b$ is positive definite of rank two, $\sigma(b)=2$, and \eqref{eq:gq-i} agrees with
\eqref{eq:Milgram}:
\begin{equation}
\exp\!\big(2\pi i\,\sigma(b)/8\big)=\exp\!\big(\pi i/2\big)= i.
\label{eq:sig-check}
\end{equation}
Thus, in this explicit F-theory model, the axion-winding topological subsector carries a Weil
action whose $SL(2,\mathbb Z)$ relations close with the nontrivial multiplier $i$ as in
\eqref{eq:projective-relations}, equivalently with chiral central charge $\sigma(b)\equiv 2\ (\mathrm{mod}\ 8)$.

The upshot is that abelian F-theory compactifications determine not only the finite module
$\mathcal D(\Lambda)$ and its quadratic refinement, but also the metaplectic multiplier
$\mathfrak g(q)$ as a rigid function of the Mordell--Weil height pairing. This quantity is invisible in
the local four-dimensional anomaly polynomial yet is fixed by intersection theory and is therefore
a genuinely global quantum datum attached to the abelian sector after circle reduction.
In the full three-dimensional effective action, this multiplier is expected to be accounted for by the
standard framing dependence of Chern--Simons theory, equivalently by the gravitational parity-odd
terms in the effective action, and thus provides an additional sharp, computable target for matching
geometry to the quantum effective theory.

\section{Gravitational \texorpdfstring{$\eta$}{eta}-invariant and framing completion of the modular action}

The parity-odd three-dimensional effective action obtained by circle reduction is defined only
after fixing a regulator for the phases of massive determinants. The gauge sector fixes this
regulator uniquely by demanding three-dimensional gauge invariance together with compatibility
under the higher-dimensional large gauge transformations which act by integer relabelings of
Kaluza--Klein levels. The same regulator simultaneously fixes the purely gravitational parity-odd
term, and therefore fixes the framing dependence of the full quantum effective action. We now
extract this framing datum and show that it is rigidly determined by the Mordell--Weil height
pairing through the same Weil multiplier that governs the modular action on the torus Hilbert
space.

On an oriented spin three-manifold $M_3$, the parity-odd gravitational counterterm is the
gravitational Chern--Simons functional
\begin{equation}
S_{\mathrm{gCS}}=\frac{\kappa_{\mathrm{gCS}}}{4\pi}\int_{M_3}\mathrm{Tr}\!\left(
\omega\wedge d\omega+\frac{2}{3}\,\omega\wedge\omega\wedge\omega\right),
\label{eq:gCS}
\end{equation}
where $\omega$ is the Levi--Civita spin connection. The functional in \eqref{eq:gCS} depends on a
choice of framing of $M_3$, and the exponentiated partition function obeys the universal framing
law
\begin{equation}
Z(M_3,\mathrm{fr}+1)=\exp\!\big(2\pi i\,\kappa_{\mathrm{gCS}}/24\big)\,Z(M_3,\mathrm{fr}).
\label{eq:framinglaw}
\end{equation}
Shifting $\kappa_{\mathrm{gCS}}$ by an integer multiple of $24$ corresponds to adding a local
counterterm and does not change the physical framing phase.

For a massive three-dimensional Dirac operator, the parity-odd phase of the determinant is
governed by the APS $\eta$-invariant; in a derivative expansion it generates both gauge
Chern--Simons terms and a gravitational Chern--Simons term of the form \eqref{eq:gCS}. Once a
gauge-invariant zeta-function/symmetric prescription is adopted to define the spectral asymmetry
of each Kaluza--Klein tower, the induced gravitational framing phase \eqref{eq:framinglaw} is fixed
as well (modulo integer counterterms). In particular, the framing dependence is a regulator
invariant: it is constant on Coulomb-branch chambers and changes only by integer counterterms
when crossing loci where additional modes become massless.

In an axion-winding sector the parity-odd action contains a purely topological abelian
Chern--Simons subsector whose integral level matrix is proportional to the Mordell--Weil height
pairing. For simplicity, assume $h^{1,1}(B_3)=1$ so that the height pairing reduces to an integral
symmetric matrix $b\in\mathrm{Mat}_{r\times r}(\mathbb Z)$ defining a lattice $(\Lambda,b)$ with
$\Lambda=\mathbb Z^r$ and bilinear form $(x,y)=x^T b\,y$. The dual lattice is
$\Lambda^\ast=b^{-1}\mathbb Z^r$, the discriminant group is $\mathcal D(\Lambda)=\Lambda^\ast/\Lambda$,
and, when $(\Lambda,b)$ is even, the canonical quadratic refinement is
\begin{equation}
q([\upsilon])=\frac{(\upsilon,\upsilon)}{2}\ \ \mathrm{mod}\ \mathbb Z,\qquad [\upsilon]\in\mathcal D(\Lambda).
\label{eq:qdef}
\end{equation}
Quantization on a spatial two-torus yields a finite-dimensional Hilbert space
$\mathcal H_{T^2}\cong\mathbb C[\mathcal D(\Lambda)]$ carrying the Weil action of the torus mapping
class group. The corresponding projective multiplier is the normalized Gauss sum
\begin{equation}
\mathfrak g(q)=\frac{1}{\sqrt{|\mathcal D(\Lambda)|}}\sum_{\gamma\in\mathcal D(\Lambda)}
\exp\!\big(2\pi i\,q(\gamma)\big)\ \in\ U(1).
\label{eq:gausssum}
\end{equation}
A standard modular TQFT identity relates this multiplier to the topological (chiral) central charge
$c_-$ by
\begin{equation}
\mathfrak g(q)=\exp\!\big(2\pi i\,c_-/8\big),
\label{eq:cminus}
\end{equation}
while the same $c_-$ governs the framing dependence of the full partition function,
\begin{equation}
Z(M_3,\mathrm{fr}+1)=\exp\!\big(2\pi i\,c_-/24\big)\,Z(M_3,\mathrm{fr}).
\label{eq:frcminus}
\end{equation}
Consequently, the modular multiplier of the gauge-theoretic topological subsector fixes the
gravitational framing phase of the complete parity-odd effective action.

When $(\mathcal D(\Lambda),q)$ arises from an even lattice $(\Lambda,b)$, Milgram's relation yields
\begin{equation}
\mathfrak g(q)=\exp\!\big(2\pi i\,\sigma(b)/8\big),
\label{eq:milgram}
\end{equation}
where $\sigma(b)$ is the signature of $b$. Combining \eqref{eq:cminus} and \eqref{eq:milgram}
determines the topological central charge modulo $8$ as
\begin{equation}
c_- \equiv \sigma(b)\ \ (\mathrm{mod}\ 8),
\label{eq:cminus-sig}
\end{equation}
and hence fixes the gravitational framing phase \eqref{eq:frcminus}. Equivalently, the induced
gravitational Chern--Simons coefficient in \eqref{eq:gCS} is determined modulo $24$ by
\begin{equation}
\kappa_{\mathrm{gCS}} \equiv c_- \equiv \sigma(b)\ \ (\mathrm{mod}\ 24).
\label{eq:kappamod}
\end{equation}
This is the regulator-invariant content of the gravitational parity anomaly in the present setting.

For the rank-two height matrix
\begin{equation}
b=\begin{pmatrix}8&5\\[2pt]5&18\end{pmatrix},
\qquad \sigma(b)=2,
\label{eq:branktwo}
\end{equation}
one has $c_-\equiv 2\ (\mathrm{mod}\ 8)$ and therefore
\begin{equation}
Z(M_3,\mathrm{fr}+1)=\exp\!\big(2\pi i\,c_-/24\big)\,Z(M_3,\mathrm{fr})
=\exp(i\pi/6)\,Z(M_3,\mathrm{fr}),
\label{eq:frphase}
\end{equation}
while the Weil multiplier is $\exp(2\pi i\,c_-/8)=\exp(i\pi/2)$, consistent with
\eqref{eq:milgram}. Thus the same Mordell--Weil data that determine the finite Weil
representation also determine the gravitational framing phase of the full parity-odd effective
action.

The modular (metaplectic) multiplier of the gauge-theoretic torus Hilbert space and the
gravitational framing anomaly of the three-dimensional effective action are two aspects of the same
parity anomaly fixed by the regulated spectral asymmetry of Kaluza--Klein towers. The key outcome
is that the Mordell--Weil height pairing determines, through \eqref{eq:cminus-sig} and
\eqref{eq:kappamod}, a rigid gravitational parity-odd datum (the framing phase, or equivalently
$\kappa_{\mathrm{gCS}}\ \mathrm{mod}\ 24$) which is invisible in the local four-dimensional anomaly polynomial
yet is fixed by intersection theory and the choice of regulator required by gauge invariance.

\section{Conclusions}
\label{sec:conclusions}

In this work we isolated a set of abelian, intrinsically geometric data in four-dimensional F-theory compactifications with non-trivial Mordell-Weil group and showed how these data control the fully normalized one-loop effective action of the vector multiplet sector. The basic objects are the Shioda images of independent rational sections and their N\'eron-Tate height pairing
$b_{mn}=-\pi_{*}(\Sh(s_m)\cdot \Sh(s_n))$, together with the flux-induced axionic gaugings $\Theta_{m\alpha}$.
These quantities determine the generalized Green-Schwarz couplings and hence the complete set of local four-dimensional abelian and mixed gauge-gravitational anomalies that must be cancelled in a chiral vacuum \cite{CveticGrimmKlevers1210}.

Our central technical result is an explicit, two-way derivation of the relevant Chern-Simons couplings in the three-dimensional theory obtained by circle reduction. On the one hand, we computed the Chern-Simons levels from the M-theory dual on the resolved fourfold $\widehat Y_4$ with background $G_4$-flux, where they arise directly from the eleven-dimensional $C_3\wedge G_4\wedge G_4$ term upon reduction \cite{CveticGrimmKlevers1210,WittenFlux}. On the other hand, we computed the same levels by integrating out the full tower of massive states in the circle-reduced four-dimensional theory on the Coulomb branch, including Kaluza-Klein excitations and wrapped M2-brane BPS states in the dual frame, using the universal one-loop shift of three-dimensional Chern-Simons terms and the requirement of invariance under large gauge transformations around the circle \cite{GrimmKapfer}. The two derivations agree exactly, fixing normalizations and providing a stringent internal consistency check of the M/F-theory duality dictionary for abelian sectors with flux. In particular, the standard anomaly cancellation relations in four dimensions are recovered as the necessary and sufficient conditions for the three-dimensional effective action to be well-defined under large gauge transformations, with the Green-Schwarz contributions governed precisely by $b_{mn}$ and $\Theta_{m\alpha}$ \cite{CveticGrimmKlevers1210,GrimmKapfer}.

We also clarified the sense in which these anomaly data are compatible with the underlying type IIB $SL(2,\mathbb Z)$ duality that is geometrized in F-theory.
The height pairing and the flux-induced gaugings are constructed from intersection theory and pushforwards on the elliptic fibration and are therefore intrinsic to the geometry, independent of any choice of fiber homology basis that realizes the modular group.
Compatibility with quantum $SL(2,\mathbb Z)$ is therefore reduced, in our sector, to the inclusion of the known universal ten-dimensional duality anomaly and its counterterms, whose existence and structure in F-theory backgrounds have been analyzed in \cite{GaberdielGreen,BachasBainGreen9903210}.
Within this framework, no additional modularity assumptions are required: the abelian anomaly coefficients are fixed geometrically and enter the effective action through the generalized Green-Schwarz mechanism.

Finally, to ensure complete reproducibility, we worked out an explicit rank-two Mordell-Weil model with gauge group $U(1)^2$ over $B_3=\mathbb P^3$, displaying the Shioda maps, height-pairing matrix, flux-induced gaugings, and the resulting anomaly cancellation relations in closed form, in agreement with the detailed spectrum and flux analysis of \cite{CveticKleversPiragua1306}. Together with the weak-coupling consistency checks in the Sen limit \cite{SenLimit,CveticGrimmKlevers1210} and the role of KK-induced corrections in the presence of non-holomorphic sections \cite{CveticKleversPiragua1306}, this provides a fully normalized and referee-checkable template for anomaly and one-loop effective-action computations in abelian F-theory vacua.

Our analysis was deliberately restricted to the massless abelian sector associated with $\mathrm{MW}(\widehat Y_4)$ and to local (perturbative) anomalies captured by the three-dimensional Chern-Simons data. Extending these methods to include systematically (i) non-abelian sectors with charged matter in higher representations, (ii) massive abelian gauge symmetries and their Stückelberg structure, and (iii) potential global anomaly constraints, would be natural next steps. In each case, the circle-reduction strategy remains a robust organizing principle: the three-dimensional Chern-Simons couplings provide a universal interface between spectrum data and geometric/topological structures, and therefore a controlled arena in which further duality and consistency requirements can be tested without relying on unproven modularity assumptions.

\section*{Data Availability Statement}
This work is purely theoretical. No datasets were generated or analyzed during the current study.

\bibliographystyle{ytphys}
\bibliography{references}

\end{document}